\documentclass[conference]{IEEEtran}
\IEEEoverridecommandlockouts
\usepackage{cite}
\usepackage{amssymb,amsfonts}
\usepackage{algorithmic}
\usepackage{graphicx}
\usepackage{textcomp}
\usepackage{subfigure}
\usepackage{booktabs} 
\usepackage{diagbox}
\usepackage{epsfig}
\usepackage{bm}
\usepackage{indentfirst}
\usepackage[linesnumbered,ruled,vlined]{algorithm2e}
\usepackage{enumerate}
\usepackage{makecell}
\usepackage{amsmath}
\usepackage{amsthm}
\usepackage{verbatim}
\usepackage{caption}
\usepackage{url}
\usepackage{cite}
\usepackage{xcolor}
\def\BibTeX{{\rm B\kern-.05em{\sc i\kern-.025em b}\kern-.08em
		T\kern-.1667em\lower.7ex\hbox{E}\kern-.125emX}}

\newtheorem{definition}{Definition}
\newtheorem{exam}{Example}
\newtheorem {strategy}{Strategy}
\newtheorem{theorem}{Theorem}
\newcommand{\stitle}[1]{\vspace{1ex} \noindent{\bf #1}}

\begin{document}
	\SetKwProg{When}{when}{ do}{}
	\SetKwProg{Function}{function}{ }{}
	
	\title{Concurrency Protocol Aiming at High Performance of Execution and Replay for Smart Contracts}

	\author{\IEEEauthorblockN{Shuaifeng Pang, \,\, Xiaodong Qi, \,\, Zhao Zhang,\,\, Cheqing Jin,\,\, Aoying Zhou}
		\IEEEauthorblockA{\textit{School of Data Science and Engineering} \\
		\textit{East China Normal University}\\
		Shanghai, China \\
		\{sfpang, xdqi\}@stu.ecnu.edu.cn, \{zhzhang, cqjin, ayzhou\}@dase.ecnu.edu.cn}
	}

	\maketitle
	
	\begin{abstract}
		Although the emergence of the programmable smart contract makes blockchain systems easily embrace a wider range of industrial areas, how to execute smart contracts efficiently becomes a big challenge nowadays. Due to the existence of Byzantine nodes, the mechanism of executing smart contracts is quite different from that in database systems, so that existing successful concurrency control protocols in database systems cannot be employed directly. Moreover, even though smart contract execution follows a two-phase style, i.e, the miner node executes a batch of smart contracts in the first phase and the validators replay them in the second phase, existing parallel solutions only focus on the optimization in the first phase, but not including the second phase. 
		
		In this paper, we propose a novel efficient concurrency control scheme which is the first one to do optimization in both phases. Specifically, (i) in the first phase, we give a variant of OCC (Optimistic Concurrency Control) protocol based on {\em batching} feature to improve the concurrent execution efficiency for the miner and produce a schedule log with high parallelism for validators. Also, a graph partition algorithm is devised to divide the original schedule log into small pieces and further reduce the communication cost; and (ii) in the second phase, we give a deterministic OCC protocol to replay all smart contracts efficiently on multi-core validators where all cores can replay smart contracts independently. Theoretical analysis and extensive experimental results illustrate that the proposed scheme outperforms state-of-art solutions significantly.
		
	\end{abstract}

	
	\section{Introduction} \label{sec:intro}

As a kind of distributed ledger shared by many non-trusted parties, blockchain technology, such as Bitcoin, Ethereum\cite{ethereum} and Hyperledger Fabric\cite{hyperledgerfabric}, has gained lots of attention and interest from public and academic communities. Programmable smart contracts, defining multiple functions to describe any business logic, promote utilization of blockchain technology for the traditional industry. 
The notion of smart contracts was conceived by Nick Szabo\cite{szabo1997formalizing} as a kind of digital vending machine in 1993. Nowadays, smart contracts can be written in several high-level programming languages, such as Solidity in Ethereum and Go in Hyperledger Fabric\cite{hyperledgerfabric}. Generally, a smart contract is invoked by a blockchain client via a transaction, along with appropriate parameters. In other words, the term \textit{transaction} is referred to as an event in which a specific smart contract is invoked.
Note that this kind of transaction also satisfies ACID (Atomic, Consistency, Isolation, and Durability) properties, like those supported in database systems\cite{dinh2018untangling}. 

\begin{figure}[hp]
	\begin{center}
		\vspace{-0.2cm}  
		\setlength{\abovecaptionskip}{0.1cm}   
		\setlength{\belowcaptionskip}{-0.4cm}   
		\includegraphics[scale=0.74]{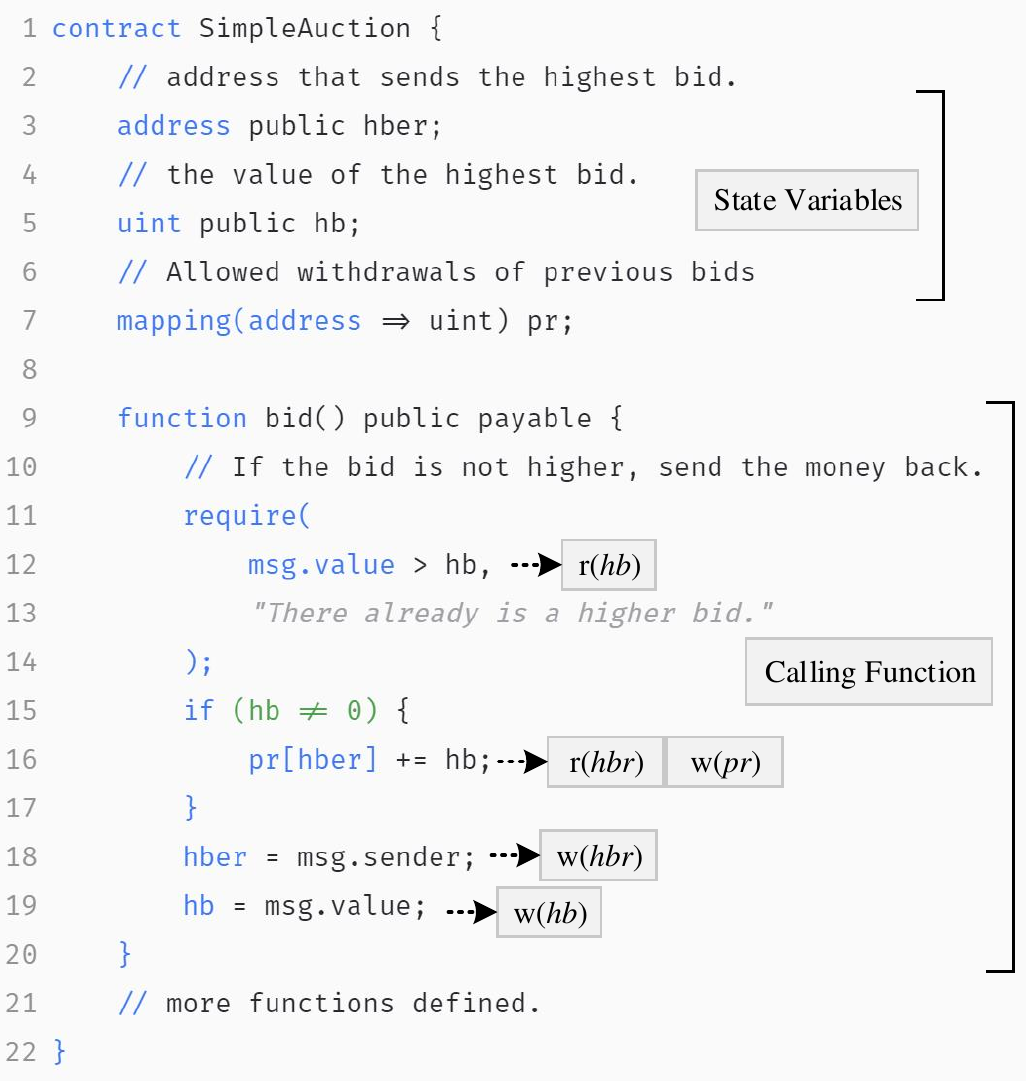}
		\caption{Open auction smart contract}
		\label{img:auction}
	\end{center}
\end{figure}

For illustration purpose, Figure \ref{img:auction} uses a simple public auction smart contract, written in Solidity, to describe a scenario where anyone can send her bid during the bidding period. Initially, three state variables are declared: address that sends the highest bid ($hber$), the highest bid ($hb$), and $pr$ that binds the bidders to their bids for refunding purpose (lines 2-7). Subsequently, function $bid()$ is declared to store the highest bidder and her bid (lines 9-20). Every bidder calls $bid()$ to submit her bid by sending a transaction to contract address. Once the highest bid rises, the former highest bidder gets her money back by calling a withdraw function (not listed in this paper due to lack of sufficient space). The \textit{msg} variable stores the sender's information such as address and balance.

\subsection{Characteristics and Challenges}

Although simple, executing smart contracts in a serial manner is exactly inefficient. 
Hence, it is critical to devise concurrent execution protocol to pursuit better performance. However, the concurrency control protocols used in traditional databases cannot be  applied to blockchain directly due to the following reasons.
\vspace{-0.1cm}
\begin{itemize}
	\setlength{\itemsep}{0.5pt}
	\setlength{\parsep}{1pt}
	\setlength{\parskip}{0cm}
	\item \textbf{Byzantine fault tolerance}. Note that the \textsf{fail-stop} assumption no longer holds in a Byzantine environment. Since Byzantine (i.e., arbitrary) nodes in blockchain systems may send false messages to other nodes and do malicious acts, the smart contract transactions need to be re-executed at all nodes rather than accepting execution results from miner directly.
	\item \textbf{Full replication data distribution}. As blockchain system adopts full replication mechanism, i.e., each node holds a complete copy of data, the concurrency control protocol in blockchain system should ensure the execution results on each replica deterministic, which makes the concurrency control protocols in traditional database system that merely guarantee serializability insufficient.
	\item \textbf{Turing-complete programming language}. 
	Smart contracts is often written in a Turing-complete language, which means conflict relationships cannot be determined till run-time.
	\item \textbf{Batching process}. A significant property of blockchain systems is that transactions are arriving in a batch style, i.e., no transaction will be processed until the block is fulfilled. The latency of batching in blockchain can be negligible compared to that in batching database systems which tends to take extra time to collect a batch of transactions. Moreover, existing batching techniques are not suitable for blockchain because they always strike a balance between throughput and latency, which is not the focus of blockchain systems.
\end{itemize}

\vspace{-0.1cm}


Hence, it is challenging to devise efficient solutions to execute smart contracts in parallel.
Since smart contract transactions must be executed by all nodes to keep state data consistent, a so-called {\em two-phase} framework is commonly adopted for concurrency control of smart contracts, as shown in Figure \ref{img:framework} \cite{anjana2018efficient,dickerson2017adding,zhang2018enabling}.
During the first phase (mining phase), miner executes transactions in parallel, and then transfers concurrent schedule log to validators. During the second phase (validation phase), each validator replays all transactions deterministically and verify whether miner is malicious or not.

Dickerson et al. proposed a solution in which the miner concurrently executes transactions using abstract locks and inverse logs to discover a serializable schedule \cite{dickerson2017adding}. Schedule logs are represented by a directed acyclic graph (\textit{happen-before} graph) to help validator recognize transactions without conflicts and execute them concurrently with fork-join\cite{lea2000java} method. Such transaction-level schedule logs make replay quite inefficient in validators.
And the abstract lock in this concurrency control is pessimistic in nature with poor scalability. Anjana et al. replaced pessimistic lock with OCC (Optimistic Concurrency Control), a cheap protocol that scales up well for low-conflict workloads \cite{anjana2018efficient}. However, once the workload has high read-write conflicts, OCC will cause high abort rate which strongly limits throughput. Zhang et al. presented a fine-grained concurrency control for validators by recording the write set of every transaction, which makes all contention relationships pre-determined \cite{zhang2018enabling}. Because the proposed mechanism, called \textit{MVTO}, uses write chain to resolve conflicts at run-time, the communication overhead and storage consumption brought by write sets further reduce the overall throughput. Moreover, all existing work ignores the possibility that the schedule log is tampered with by miner and relies on the default verification mechanism of blockchain systems.

\begin{figure}[!hbt]
	\begin{center}
		\includegraphics[scale=0.39]{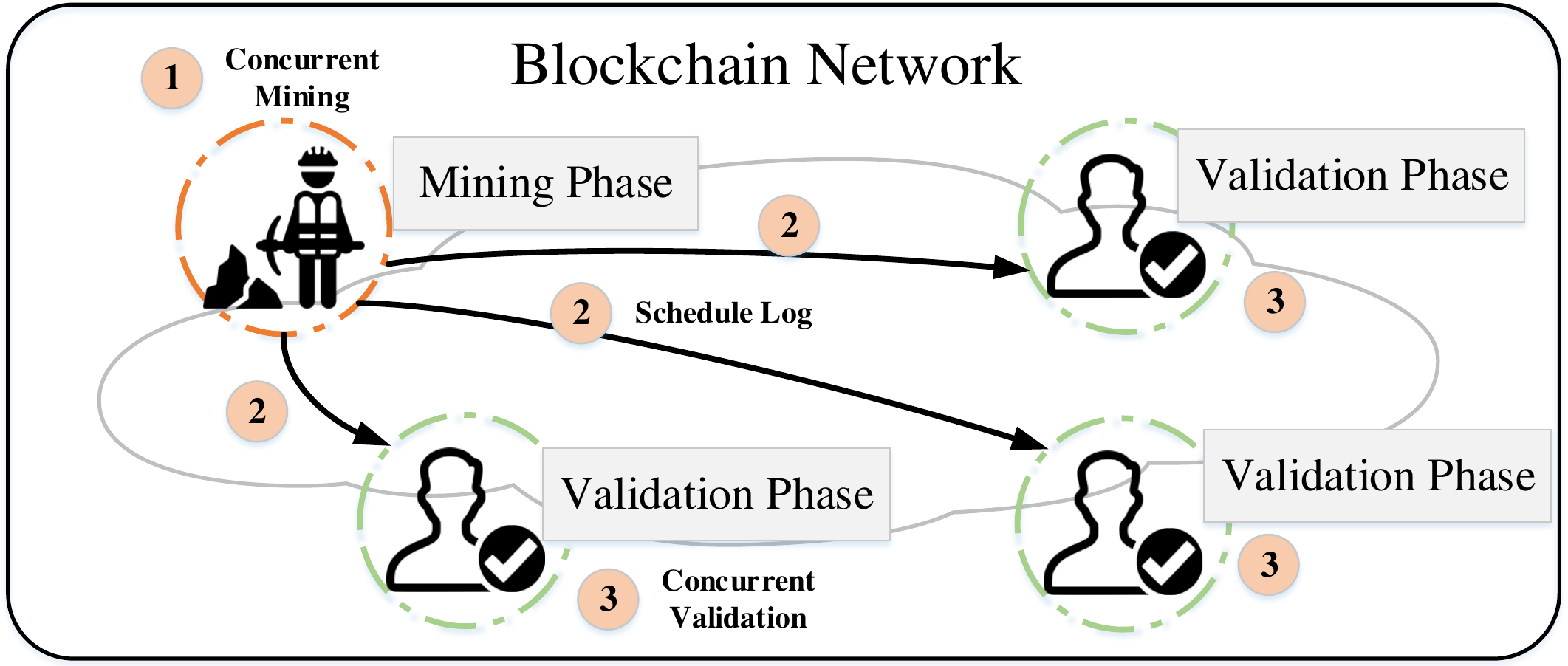}
		\caption{Two-phase concurrent execution framework.}
		\label{img:framework}
	\end{center}
\end{figure}

To sum up, existing work has following three drawbacks: (i) Miner tends to adopt mature and proven concurrency control protocols, like two-phase locking (2PL) and OCC in the database area, which only benefits miner itself. (ii) The granularity of concurrent schedule log is either too coarse, or too fine, which affects communication overhead and replay throughput deeply. (iii) The verification mechanism of validators depends on the default setting of blockchain systems, which checks the state Merkle-root after completing the state transition. This prevents a malicious miner being detected rapidly during the execution. 

In view of the above three limitations, we propose a novel concurrency control protocol for smart contract execution, aiming at boosting performance for miner and improving replay efficiency for validators. Specifically, it includes three aspects: First, a variant of OCC combined with transaction batching is proposed for the miner, where aborted transactions are carefully picked in validation phase of OCC; Second, an appropriate granularity of schedule log is determined, because coarse-grained schedule log like \textit{happen-before} graph\cite{dickerson2017adding} causes low throughput in validation phase, and fine-grained one such as \textit{write chain}\cite{zhang2018enabling} results in massive communication cost. Third, our deterministic concurrent replay scheme based on medium-grained schedule log allows validators to replay the same schedule in a concurrent and deterministic manner.

\subsection{Contributions}
More specifically, we claim the following contributions in this study:
\begin{itemize}
	\setlength{\itemsep}{0pt}
	\setlength{\parsep}{1pt}
	\setlength{\parskip}{0pt}
	\item Design an effective OCC variant according to transaction batching feature for miner, which can provide higher parallelism for miner and faster replaying speed for validators at the same time. In other words, we take the optimization of replaying performance of validators into consideration from the beginning (mining phase).
	
	\item Devise medium-grained concurrent schedule log to validators. Even though our proposed \textit{partitioned transaction dependency graph} maintains high concurrent degree, the communication cost is cheap.
	
	\item Propose a deterministic and conflict-free concurrency control protocol for validators. We also design a method that embeds verification scheme into the proposed protocol to quickly detect malicious tampering.
	
	
	\item A prototype implemented in Java. We integrate the above techniques into this prototype and measure the system performance under a standard benchmark. Experimental results show the superiority of the proposed methods.

\end{itemize}

\subsection{Organization} 
The rest of paper is organized as follows. Section \ref{sec:related} reviews some related work about latest approaches applied to concurrent smart contract execution and other concerned techniques. Section \ref{sec:problem} formally defines the problems in this paper. In Section \ref{sec:framework}, we explain our parallel two-phase execution scheme in detail. Section \ref{sec:analysis} analyzes the overall cost of the proposed protocols. The experimental evaluations are reported in section \ref{sec:evaluation}. Section \ref{sec:conclusion} concludes this paper.

%
\section{Related Work}\label{sec:related}
We review recent researches close to our work in this section.


\stitle{Concurrency control protocol}. Concurrency control in DBMS has been actively studied for more than 30 years. Generally these works are grouped into two kinds, one is pessimistic, and the other is optimistic. The most frequently used pessimistic scheme to ensure serializability is two-phase locking (2PL) \cite{eswaran1976notions}. In contrast to the pessimistic lock-based protocol, Kung and Robinson propose a validation-based, non-locking optimistic concurrency control scheme, or in short OCC \cite{kung1981optimistic}. As a validation based protocol, every transaction goes through three phases: first comes a read phase where the transaction reads data item directly from storage and writes to a private location. Then the transaction enters validation phase. If it passes validation, the transaction writes back its updates. Wang et al. propose an adaptive CC which combines with both OCC and lock \cite{wang2016mostly}.

\stitle{Concurrent smart contract execution}. Smart contract which is sequential programs stored on blockchain can be triggered by transaction sent by clients. While a pile of works try to improve the performance of blockchain system on consensus layer, there exist some works on adding concurrency to smart contract execution. Dickerson et al. present a solution to permit miner and validators to execute smart contracts concurrently \cite{dickerson2017adding}. Every smart contract invocation can be treated as a speculative action so that miner can discover a serializable schedule using lock-based STM and publish it to the blockchain. Validators who receive a new block can replay the same execution deterministically. Zhang et al. propose a method which can employ any concurrency control mechanism that produces a conflict-serializable schedule in mining phase \cite{zhang2018enabling}. Validators use MVTO protocol with the help of write sets provided by miner to re-execute transactions. Anjana et al. replace the pessimistic protocol with OCC and propose a decentralized way in validation phase \cite{anjana2018efficient}. Sergey and Hobor explore similarity between multi-transactional behaviors of smart contracts in Ethereum and shared-memory concurrency problem \cite{sergey2017concurrent}. These approaches consider the problem separately while ours takes the overall interest of miner and validators into account from the beginning.

\stitle{Batching and Determinism}.
Batching processing is commonly used to improve performance. Ding and Kot utilize transaction batching and reordering techniques to improve OCC \cite{ding2018improving}. Santos et al. apply batching technique to boost throughput of Paxos \cite{santos2012tuning}. Batching is also a fundamental feature of blockchain systems. We combine this feature with OCC to reduce abort rate.
Determinism is also a concerned topic in concurrent execution. Bocchino et al. argue that parallel programming must be deterministic by default\cite{bocchino2009parallel}. And also several approaches are brought up by Bocchino. Vale et al. present a deterministic transaction execution system in the context of Transaction Memory (TM) \cite{vale2016pot}.
Recall that the working style between blockchain systems and database systems is significantly different, so that it is necessary to devise novel solutions.
	\section{Problem Statement}\label{sec:problem}

In this section, we will formalize the key issues in this study. Table \ref{tab:symbol} lists the notation used throughout this paper. 

\begin{table}[htbp]
	\caption{The notation used in this paper}
	\label{tab:symbol}
	\centering
	{\small 
		\begin{tabular}{c|l|c|l}
			\hline
			Symbol & Description & Symbol & Description\\
			\hline
			$G$ & garph & $\Pi$ & parallelism of graph \\
			$|V|$ & \# of vertices & $\rho$ & commit ratio \\
			$|E|$ & \# of edges & $\tau$ & workload threshold \\
			$\omega(v)$ & weight of $v$ & $B$ & a batch of txs \\
			$R_j^i$ & consistent read set & $RS(T_i)$ & read set of $T_i$ \\
			$P$ & a partition of graph & $WS(T_i)$ & write set of $T_i$\\
			$D$ & density of graph & $\mathcal{O}$ & serialization order \\
			$c(e)$ & communication cost & & \\
			\hline
		\end{tabular}
	}
\end{table}

As in database system, multiple smart contract transactions running concurrently can cause data races leading to inconsistent final state in the blockchain.         
Hence, data conflicts need to be resolved during run-time to ensure consistency.
Serialization graph or conflict graph($CG$)\cite{adya1999weak}\cite{cahill2008serializable} has long been adopted in concurrency control to captures conflict relationship among concurrent transactions, in which vertices are smart contract transactions, and edges represent \textit{read-write} conflict dependencies between smart contract transactions. Note that \textit{write-write} conflict dependencies need not to be tracked in $CG$ because each transaction maintains its own write set in OCC protocol.     
For simplicity, smart contract transaction is abbreviated as transaction hereafter.

\begin{definition}[\small{Conflict Graph, CG}] \label{def:cg}
	A CG is a directed graph $G=(V, E)$, where $V=\{T_1, T_2, T_3, \ldots, T_n\}$, $E=\{(T_i, T_j)|i \neq j$\}. We say there is a read-write conflict edge from $T_i$ to $T_j$ if $RS(T_i) \cap WS(T_j) \neq \emptyset$ holds where $RS(T)$ and $WS(T)$ denote read set and write set respectively.
\end{definition}

Suppose three concurrent transactions calling function \textit{bid} of smart contract shown in Figure \ref{img:auction} are abstracted as below. Both $T_1$ and $T_2$ manage to raise the highest bid. $T_3$ sends a bid with the value no greater than the current highest bid ($hb$), so it returns directly.
\[
\begin{array}{l}
T_1: \textsf{r}(hb)\textsf{r}(hber)\textsf{w}(pr)\textsf{w}(hber)\textsf{w}(hb) \\ \vspace{0.1cm}
T_2: \textsf{r}(hb)\textsf{r}(hber)\textsf{w}(pr)\textsf{w}(hber)\textsf{w}(hb) \\ \vspace{0.1cm}
T_3: \textsf{r}(hb)
\end{array}
\]
By the above definition, we check read sets and write sets between any two transactions and construct a conflict graph shown in Figure \ref{img:graph:a}. As usual in traditional concurrency control theory, the absence of a cycle in $CG$ proves that the schedule is serializable. If $CG$ is acyclic, a serialization order $\mathcal{O}$ can be acquired by repeatedly committing a transaction without any edge under a topological order. Otherwise, we need to abort several transactions to make the graph acyclic, i.e., no dependency relationship among all the remaining transactions. 

However, $CG$ cannot capture dependencies among those rollback transactions and commit transactions. So, we define a transaction dependency graph ($TDG$) to represent the final schedule of a batch transactions $B$. $TDG$ is generated on the basis of a serializable $CG$ by adding vertices and updating edges when committing transactions according to the serialization order $\mathcal{O}$. 


\begin{definition}[\small{Transaction Dependency Graph, $TDG$}] \label{def:tdg}
	A $TDG$ is a DAG (directed acyclic graph) $G=(V, E)$, where $V=\{T_1, T_2, T_3, \ldots, T_n\}$, $E=\{(T_i,T_j,R_j^i)|1 \leq i \neq j \leq n\}$ and $R_j^i$ records all values that $T_j$ reads from $T_i$. 
\end{definition}

An example $TDG$ is presented in Figure \ref{img:graph:b}. After aborting transaction $T_2$ in the example $CG$, we can commit the remaining two transactions in a serialization order ($T_1 \rightarrow T_3$). If re-executing $T_2$, a dependency edge from $T_1$ to $T_2$ is included in the final $TDG$ with a consistent read set $R_2^1=\{hb, hber\}$. Edges of a $TDG$ indicate \textit{read-from} relationships, i.e., the data transferred from one transaction to another. The weight of each vertex $T_i$ is defined as the execution time and is denoted as $\omega(T_i)$. Hereafter, we use the terms vertex and transaction interchangeably. Edges in $TDG$ correspond to both the precedence constrains and communication message containing \textit{consistent read sets} among vertices. The weight of an edge $e=(T_i,T_j)$ denoted as $c(e)$ indicates the communication overhead (by the byte size of $R_j^i$). 

\begin{figure}[htbp]
	\centering
	\vspace{-0.1cm}  
	\setlength{\abovecaptionskip}{0.2cm}   
	\setlength{\belowcaptionskip}{-0.1cm}   
	\subfigure[$CG$]{
		\label{img:graph:a}
		\includegraphics[scale=0.53]{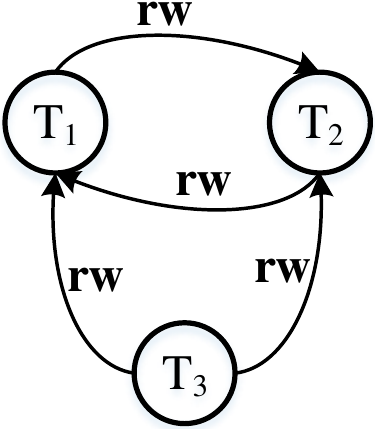}
	}
	\hspace{2cm}
	\subfigure[$TDG$]{
		\label{img:graph:b}
		\includegraphics[scale=0.53]{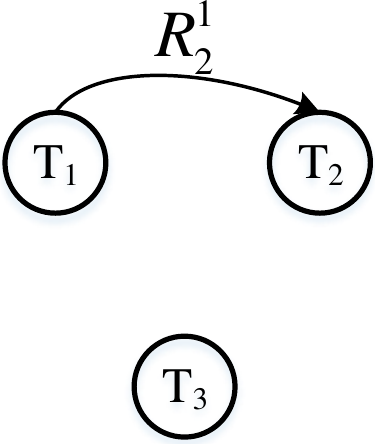}
	}
	\caption{An example of $CG$ and $TDG$.}
	\label{img:graph}
\end{figure}

As mentioned before, considering the conflict is low in blockchain system, we design our concurrency control protocol based on OCC protocol. However, the native OCC dose not guarantee a serializable schedule. Transactions that cause the violation of serializability need to be aborted to obtain an acyclic $CG$ during the validation phase of OCC. Furthermore, a careful selection of aborted transactions is required to generate a higher concurrent schedule.


Deleting vertices from a conflict graph to make it cycle-free is actually a classic feedback vertex set problem (FVS)\cite{chen2008improved}. A feedback vertex set of a directed graph is a subset of vertices whose removal makes the graph acyclic. For example, consider the $CG$ in Figure \ref{img:graph:a}. Vertex $T_3$ forms a FVS since the removal frees the graph of cycle. Once the size of the feedback vertex set becomes too large, it will downgrade the overall throughput. So we further define a commit ratio $\rho$ to represent the percentage of successfully committed transactions in $B$. Making graph acyclic with a minimal-size vertex set is the basic requirement for FVS problem. We can consider a more complex objective function which is to make the transformed $TDG$ with higher parallelism for the sake of high replay speed in validators. 

A graph with fewer edges owns higher parallelism due to less inter-conflict between transactions. Let $\Pi$ denote the sparsity of $TDG$ which indicates the concurrent degree validators can obtain during replaying, where $D$ is the density of the graph itself. The density of a graph is the ratio of edges in the graph to the maximum possible number of edges. Equation \ref{equ:density} defines $\Pi$ formally.  

\begin{equation} \label{equ:density}
\setlength{\abovedisplayskip}{-6pt}
\setlength{\belowdisplayskip}{0pt}
\Pi=1-D=1-\frac{|E|}{|V|(|V|-1)}
\end{equation}

Based on sparsity definition of $\Pi$ in Equation \ref{equ:density}, generating $TDG$ with higher sparsity and parallelism is transformed into a classic FVS problem defined by definition \ref{def:problem1}. Note that a commit ratio $\rho$ need to be set to prevent aborting overmuch transactions from happening.


\begin{definition}[FVS Problem] \label{def:problem1} 
	Given a batch of $n$ transactions $B=\{T_1,T_2,\ldots,T_n\}$, find a subset $B' \in B$, aborting $B'$ not only makes a serializable $CG$, but also minimizes the density of transformed $TDG$ with at least $\rho$ commit ratio.
\end{definition}


Generating a medium-grained schedule log represented by $TDG$ is reasonable. There are two specific reasons: (1) The theoretical concurrent degree is hard to reach due to the limited physical cores. (2) Offering all consistent read set significantly increases network overhead between miner and validators. Therefore, we come up with a partitioning way on $TDG$ which remains the parallelism as much as possible but brings much smaller communication cost.


A $\tau$-constrained partitioning of a $TDG$, $G=(V,E)$, divides $V$ into multiple disjoint subsets $\{V_1,V_2,...\}$ with the workload of each sub-graph no more than a threshold $\tau$. The weight of a sub-graph $\omega(V_i)=\sum_{T \in V_i}\omega(T)$ which is the sum of the vertex weight or transaction execution time. An edge $(T_i,T_j,R_j^i)$ is called \textit{cut edge} if $T_i \in V_p,T_j \in V_q$ and $p \neq q$. The weight of all \textit{cut edges} $c(E_c)=\sum_{e \in E_c}c(e)$. Definition \ref{def:problem2} gives a constrain on each sub-graph and accompanies the problem with an objective function based on the weights of the cut edges.

\begin{definition}[Graph partition problem] \label{def:problem2}
	Given a transaction dependency graph $G$, and an upper bound of weight $\tau$, is there a $\tau$-constrained partition $P=\{V_1,V_2,...\}$ such that $\omega(V_i)$ is no larger than $\tau$ and $c(E_c)$ is minimized?
\end{definition}


%
\section{Execution Framework}\label{sec:framework}
The two-phase execution framework is the most appropriate approach to keep state consistency among all replicas in a Byzantine environment. 
A transaction will be executed twice throughout its life cycle, once in miner and again in validators in this two-phase framework.     
The existing research work considers the concurrent problem of two phases separately, where the miner adopts the mature and proven concurrency control protocol neglecting the replay efficiency in validators.

Regarding this drawback, we propose a novel concurrency control protocol taking the optimization of two phases into account at the same time. In the first phase, we design a variant of OCC protocol utilizing the characteristic of batching and produce a medium-grained schedule log which is conducive to improve replay efficiency in validators and also reduces the communication cost. 
When receiving blocks that contain the concurrent schedule log, validators enter the second execution phase and complete an efficient replay with the help of a deterministic protocol and this schedule log. 

To be more specific, our approach is guided by three goals, including designing an efficient concurrency control protocol suitable for blockchain, determining the appropriate granularity for the schedule log, and devising a deterministic and efficient replay protocol based on the scheduling log sent by the miner for validators.

We design a variant of OCC protocol based on the natural batching feature of blockchain systems to fulfill the first goal, aiming at boosting performance for miner and improving replay efficiency for validators. And seeking a solution to get the first goal is exactly the same process of solving the FVS problem illustrated by definition \ref{def:problem1} in Section \ref{sec:problem}. More details about batching concurrency control are described in Section \ref{sec:framework:mining}.

Given that the parallelism is surplus when offering every transaction a consistent read set due to the limitation of physical cores, a practical approach is devised to lessen the communication cost between the miner and validators while remaining maximum concurrent degree during replaying in validators. The problem behind the second goal is a graph partition problem defined by Definition \ref{def:problem2} in Section \ref{sec:problem} 
Section \ref{sec:framework:mining} details the partition algorithm.


As for the third goal, we propose a deterministic OCC protocol benefiting from \textbf{partitioned TDG} to keep consistency among replicas and also improve CPU utilization. We describe the concurrent validator scheme in Section \ref{sec:validation}.


\subsection{Concurrency Control Protocol in Mining Phase}\label{sec:framework:mining}

\stitle{Batching and reordering  transactions in OCC's validation phase} 

As a processing unit in blockchain systems, each block is basically a batch of transactions. We take advantage of this natural feature and apply it to OCC protocol. As transactions in original OCC commit randomly, and the final serialization order is only decided at commit time during the validation phase of OCC, we utilize batching technique, with which the protocol waits until all transactions finishing their read phase and selects an optimal validation order , to reduce the number of conflicts and aborts. Moreover, reordering can further improve the throughput of the miner.

Figure \ref{img:reorder} gives a simple transaction reordering example. Suppose two transactions $T_1$ and $T_2$ in a batch are executed concurrently. Figure \ref{img:reorder:a} shows that $T_1$ first reads $x$ before $T_2$ writes a new version of $x$. If without batching, $T_2$ can successfully commit while $T_1$ fails the validation and leads to an abort as Figure \ref{img:reorder:b} demonstrates. But with batching and reordering, $T_2$ can be serialized after $T_1$ (Figure \ref{img:reorder:c}). Thus, two transactions can both commit its own changes without abort.

\begin{figure}[htbp]
	\centering
	\subfigure[Two Txs]{
		\label{img:reorder:a}
		\includegraphics[scale=0.53]{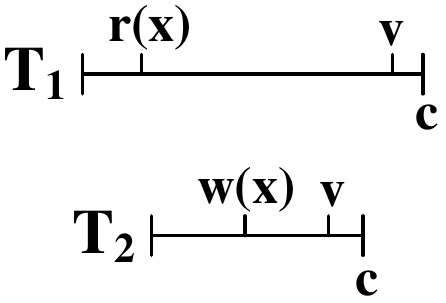}
	}
	\hfill
	\subfigure[Original OCC]{
		\label{img:reorder:b}
		\includegraphics[scale=0.53]{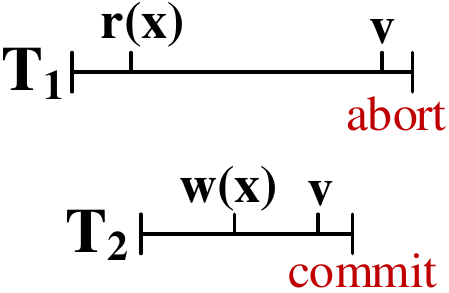}
	}
	\hfill
	\subfigure[Reordering]{
		\label{img:reorder:c}
		\includegraphics[scale=0.53]{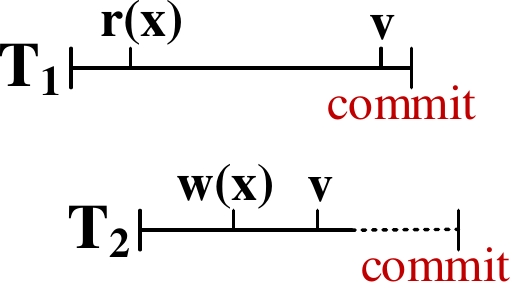}
	}
	\caption{Batching and reordering transactions in OCC's validation phase to reduce abort.}
	\label{img:reorder}
\end{figure}

Batching technique is always about the trade-off between throughput and latency. Batching transactions and reordering the commit order surely increase throughput in mining phase, so does latency. However, validators occupy the vast majority of nodes in blockchain systems. Sacrifices made by miner brings much more benefit to validators and further improves throughput of the whole system.



In order to get a serializable and high parallelism schedule, the vital problem is to find a minimal subset $B'$ of all vertices whose removal maximizes parallelism defined by Definition \ref{equ:density} of the output graph and guarantees commit ratio no less than $\rho$. Actually, if $CG$ is not acyclic, then we get an optimization problem as FVS.
Unfortunately, since FVS is an NP-hard\cite{kann1992approximability}\cite{karp1972reducibility} problem, we propose a greedy algorithm for finding $B'$ next.


\setlength{\textfloatsep}{0.15cm} 
\setlength{\floatsep}{0.15cm}
\begin{algorithm}[htbp]
	\caption{Concurrent Mining}
	\label{alg:concurrentmining}
	\KwIn{A batch of transactions $B$, commit ratio $\rho$}
	\KwOut{A transaction dependency graph $TDG$}
	{\small
		Initialize an output $TDG$; \\
		$nCommit \leftarrow 0$;\\
		$B' \leftarrow B$;\\
		\While{$nCommit < \rho |B|$}{
			$CG \leftarrow ExecuteParallel(B')$;\\
			$B' \leftarrow FindAbortTransactionSet(CG)$;\\
			$CG' \leftarrow CG \setminus B'$;\\
			$\mathcal{O} \leftarrow TopologicalSort(CG')$;\\
			\For{each $t \in \mathcal{O}$}{
				$txCommit(t)$;\\
				$nCommit \leftarrow nCommit+1$;\\
				$UpdateGraph(t, TDG)$;\\
			}
		}
		\Return{$TDG$};
	}
\end{algorithm}

Algorithm \ref{alg:concurrentmining} briefs how to enable concurrency in mining phase and generate a schedule log represented by $TDG$. Line 3 initializes $B'$ with $B$ which indicates the current aborted transaction set that needs to be re-executed. The codes inside the \textit{while loop} (Lines 4-12) firstly runs $B'$ in parallel. When all transactions of a batch reach validation phase, we construct a local $CG$ by creating one vertex per transaction, and one edge per \textit{read-write} conflict relationship. Algorithm checks whether $RS(T_i) \cap WS(T_j) = \emptyset$ to determine if there is a conflict edge from $T_i$ to $T_j$. Next, function $FindAbortTransactionSet$ computes an optimized vertex set $B'$ based on $CG$. After removing $B'$ from $CG$, the algorithm repeatedly commits transactions without any incoming edge using a topological sort. Each successful commit triggers function $UpdateGraph$ which creates a new vertex and makes use of $t$'s read set to generate edges. The algorithm loops the procedure until the commit ratio $\rho$ is satisfied.


\begin{figure*}[htbp]
	\centering
	\subfigure[Original Conflict Graph]{
		\label{img:sccabort:a}
		\includegraphics[scale=0.49]{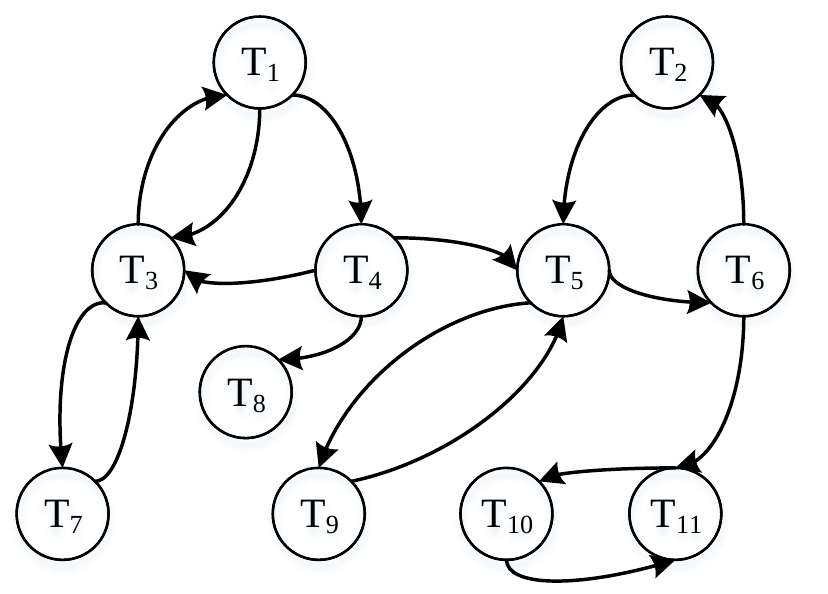}
	}
	\hfill
	\subfigure[Find all SCCs]{
		\label{img:sccabort:b}
		\includegraphics[scale=0.49]{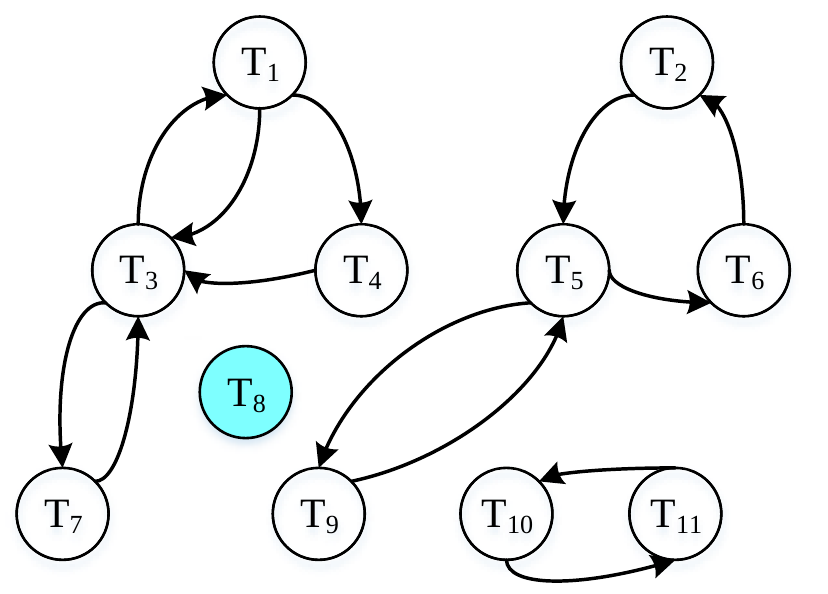}
	}
	\hfill
	\subfigure[Abort $T_3$ and $T_5$]{
		\label{img:sccabort:c}
		\includegraphics[scale=0.49]{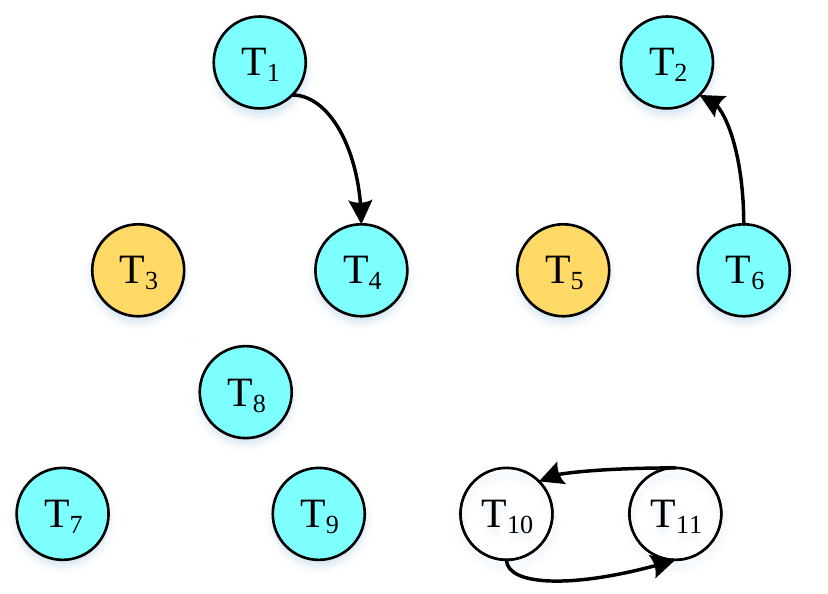}
	}
	\hfill
	\subfigure[Abort $T_{10}$]{
		\label{img:sccabort:d}
		\includegraphics[scale=0.49]{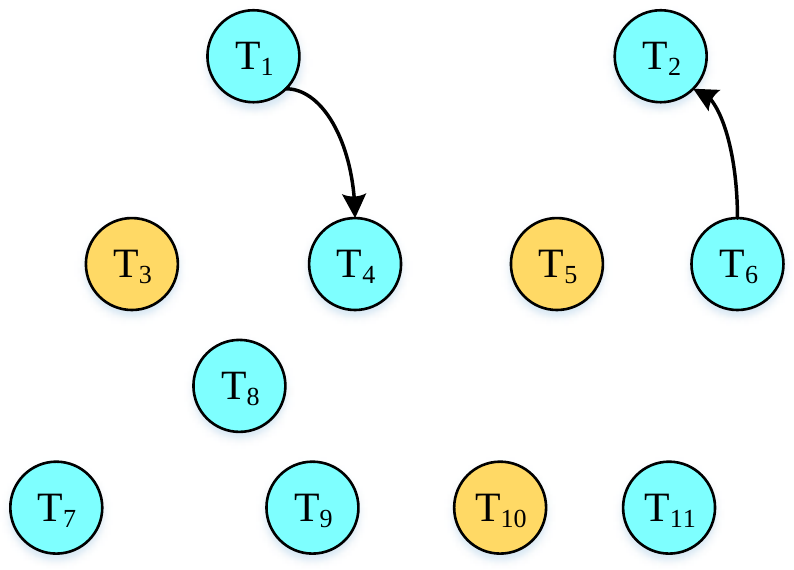}
	}
	\caption{An example of $FindAbortTransactionSet$ algorithm. Yellow vertices form the abort set, and blue ones are pruned during the algorithm.}
	\label{img:sccabort}
\end{figure*}

\stitle{Finding abort transactions}

The heart of our concurrent mining approach lies on the greedy function $FindAbortTransactionSet$. This function selects vertices that are most likely to be included in an abort set and meets our objective which is producing a $TDG$ with the largest parallelism. The heuristic rule behind Algorithm \ref{alg:aborttx} is that each strongly connected component (SCC) of the conflict graph must contain at least one cycle. We can use some previous work like Kosaraju's algorithm\cite{sharir1981strong} and Gabow's\cite{gabow2000path} to separate all SCCs in linear time complexity for further processing. 

For each SCC, function $GreedySelectVertex$ returns a subset of transactions to abort, and the output set $B'$ is the union of all subsets. Lines 7-12 in Algorithm \ref{alg:aborttx} illustrate how we recursively choose vertices to form the abort set. SCCs of size one are pruned (Line 8).
We sort all vertices within one SCC in descending order by the strategy defined by Strategy \ref{strategy1} , and greedily select the top-ranked vertex (Line 10). Line 11 removes the chosen vertex $\overline{V}$ and returns the remaining pruned graph. Then we recursively call function $GreedySelectVertex$ until sizes of all SCCs are less than 1.

\begin{strategy} [max-in \& min-out strategy] \label{strategy1}
	Suppose we have a vertex $T_a$ in $CG$ with an outgoing edge set $E_o=\{(T_a, T_i)|i \neq a\}$ and an incoming edge set $E_i=\{(T_j, T_a)|j \neq a\}$. If we abort and rerun $T_a$, then we have a corresponding incoming edge set $E_i'=\{(T_i, T_a')\}$ of $T_a'$ in the transformed $TDG$. And none of $T_a$'s incoming edge is included. As the objective is to minimize the density of $TDG$, we choose the vertex with largest in-degree and smallest out-degree to abort so that a sparser $TDG$ can be obtained.
\end{strategy}

\setlength{\textfloatsep}{0.15cm} 
\setlength{\floatsep}{0.15cm}
\begin{algorithm}[!hbt]
	\caption{FindAbortTransactionSet}
	\label{alg:aborttx}
	\KwIn{Conflict Graph $CG$}
	\KwOut{$B'$, a vertex set to be aborted}
	{\small
		$B' \leftarrow \emptyset$;\\
		$\overline{CG} \leftarrow Prune(CG)$;\\
		$SCC \leftarrow Kosaraju(\overline{CG})$;\\
		\For{each $S \in SCC$}{
			$B' \leftarrow B' \cup GreedySelectVertex(S)$;
		}
		\Return{$B'$};\\
		\Function{ $GreedySelectVertex(S)$}{
			\If{$|S.V| \le 1$}{
				\Return{$\emptyset$};
			}
			$\overline{V} \leftarrow ChooseVertexByStrategy(S)$;\\
			$\overline{S} \leftarrow Prune(S \setminus \overline{V})$;\\
			\Return{$\overline{V} \cup GreedySelectVertex(\overline{S})$};
		}
	}
\end{algorithm}


\vspace{-0.15cm}
\begin{exam}\label{exam:sccabort}
	Figure $\ref{img:sccabort}$ gives an instance of the greedy recursive algorithm that aims at maximizing the parallelism of the transaction dependency graph. Figure $\ref{img:sccabort:a}$ is the original example conflict graph with no vertex to be pruned during the pre-processing. We use Kosaraju's SCC algorithm\cite{sharir1981strong} to partition the example graph into several SCCs, and remove SCCs of size one which will be marked blue, e.g., vertex $T_8$ in Figure $\ref{img:sccabort:b}$. In the first SCC that consists of vertex $T_1$, $T_4$, $T_3$ and $T_7$, vertex $T_3$ has the largest in-degree, so we abort it, remove all relevant edges and mark it as yellow. The algorithm takes the remaining part of the first SCC as input and recursively find the next transaction to abort. We prune vertex with zero in-degree or out-degree before and during the algorithm. Here, vertex $T_1$, $T_4$ and $T_7$ are all trimmed (Figure $\ref{img:sccabort:c}$). We then move to next SCC containing vertex $T_2$, $T_5$, $T_6$ and $T_9$. We choose vertex $T_5$ ($in$-$degree=2$) as the abort transaction. The rest vertices of this SCC satisfy the condition for pruning once vertex $T_5$ is removed. We look at the last SCC containing vertex $T_{10}$ and $T_{11}$. Since vertex $T_{10}$ and $T_{11}$ have the same in-degree and out-degree, we can add either one of them to the abort set. Here we select vertex $T_{10}$. And vertex $T_{11}$ is pruned (Figure $\ref{img:sccabort:d}$). In Figure $\ref{img:sccabort:d}$, we have our final abort set $B' \leftarrow \{T_3, T_5, T_{10}\}$.
\end{exam}

\vspace{-0.15cm}
\subsection{Generate Moderate Granularity Schedule Log}\label{sec:framework:partition}

\begin{algorithm}[!hbt]
	\caption{$\tau$-Constrained Partitioned $TDG$}
	\label{alg:partition}
	\KwIn{A transaction dependency graph $G = (V,E)$, workload threshold $\tau$}
	\KwOut{A partition of $G$, $P=\{V_1,\ldots,V_k \}$}
	{\small
		$C \leftarrow \sum_{v \in V}{\omega(v)}$; \\
		$U \leftarrow C \times \tau$; \\
		$V_i \leftarrow \emptyset$; $cost \leftarrow 0$;\\
		$E' \leftarrow SortByWeight(E)$; \\
		\For{each $e \in E'$}{
			$uv \leftarrow e.getUnvisitedVertex()$;\\
			\If{$cost+\omega(uv) > U$}{
				$P.add(V_i)$; \\
				$V_i \leftarrow \emptyset$; \\
				$cost \leftarrow 0$; \\
			}
			$visit[uv] \leftarrow true$; \\
			$V_i \leftarrow V_i \cup \{uv\}$; \\
			$cost \leftarrow cost+\omega(uv)$; \\
			%
			%
		}
		\For{each $v \in V$}{
			\If{$visit[v]=false$}{
				\If{$cost+\omega(v) > U$}{
					$P.add(V_i)$; \\
					$V_i \leftarrow \emptyset$; \\
					$cost \leftarrow 0$; \\
				}
				$visit[v] \leftarrow true$; \\
				$V_i \leftarrow \cup \{v\}$; \\
				$cost \leftarrow cost+\omega(v)$; \\
			}
		}
		\Return{$\{V_1,V_2,\ldots,V_k\}$};
	}
\end{algorithm}

By now, we have a transaction dependency graph whose sparsity $\Pi$ is maximized. But the fine-grained scheduling log provided every transaction a consistent read set will cause surplus parallelism limited by physical cores and also too much communication overhead. So we intend to cut $TDG$ into even pieces so to improve CPU utilization and decrease the communication cost as well.


Since finding the optimal partition of $TDG$ that minimizes the size of all $R_j^i$ is an NP-hard\cite{kann1992approximability} problem. We propose another greedy algorithm to solve the problem described in Section \ref{sec:problem}. Our algorithm is based on a simple rule, that is edges with larger weight are preferred to be included in a part so that the probability of edges with smaller weight being cut edge is increased. 

Algorithm \ref{alg:partition} computes a balanced partition including multiple sub-graphs, each one with a workload no more than the threshold $\tau$. And the size of all $R_j^i$ is minimized as much as possible. Line 1-2 calculate an upper bound weight $U$ of each part. Line 3 initializes the temporary sub-graph and a local variable $cost$ records the current weight of $V_i$. Then we reorder all edges of $E$ by their weight in descending order (Line 4). Once having a sorted edge set $E'$, the algorithm traverses each edge $e$ to compute a partition with minimal edge cut (Line 5-13). Basically, if $e$ connects at least one unvisited vertex $u$, then we add it to the current sub-graph $V_i$ (Line 12) and update $cost$ with $\omega(u)$ (Line 13). When $cost$ exceeds the upper bound $U$, we get a new part (Line 7-10). After visiting every edge of $E$, there may be some unvisited vertices with no incoming/outgoing edges. So, line 14-22 iterate over all vertices in $V$, and assign them to the appropriate sub-graph. Algorithm \ref{alg:partition} accesses all edges and vertices once, it has time complexity of $O(|V|+|E|)$.

\begin{exam}\label{exam:partition}
	Figure $\ref{img:partition}$ illustrates a simple bi-partition process after applying our partition algorithm to an example dependency graph shown on Figure $\ref{img:partition:a}$. Suppose a vertex set $V \leftarrow \{T_1, T_2, \ldots, T_{11}\}$ has a corresponding weight set $W \leftarrow \{11,13,8,6,7,12,1,9,3,2,1\}$. And a workload threshold is set to 0.5. The upper bound of each sub-graph is calculated as $73 \times 0.5=36.5$. We sort all edges by their weight as line 4 suggests. Then we start by the edge with the largest weight which is edge $e: T_8 \rightarrow T_9$. Both start vertex and end vertex of $e$ are unvisited, so we add them to $V_1$. Now the cost of $V_1$ is updated to 12. Next, we process edge $T_4 \rightarrow T_8$, edge $T_1 \rightarrow T_4$ and edge $T_1 \rightarrow T_3$ in order. After that, our first sub-graph contains vertex $T_1$, $T_3$, $T_4$, $T_8$ and $T_9$. When dealing with edge $T_2 \rightarrow T_6$, we find out that the cost of current part $V_1$ is 37 which exceeds $U$, so vertex $T_2$ and $T_6$ are included in a new sub-graph $V_2$. We now have a new sub-graph shown on Figure $\ref{img:partition:c}$. After visiting all edges, there are still two unvisited vertices $T_7$ and $T_{10}$. We add them to $V_2$ one at a time. Finally, we have a 2-way partition presented in Figure $\ref{img:partition:e}$.
\end{exam}

\subsection{Replay Schedule Log in Validation Phase}\label{sec:validation}

\stitle{A deterministic OCC protocol}

The commit order of batching OCC is not deterministic due to the interleavings are arbitrary. Since $TDG$ offers all conflict relationship between transactions and the commit order is predefined as $\mathcal{O}$, batching transactions in the second phase is unnecessary. So we propose a deterministic optimistic concurrency control protocol called DeOCC based on OCC along with \textbf{partitioned TDG} and $\mathcal{O}$. Providing multiple \textit{consistent read sets} for every sub-graph, DeOCC is a non-blocking and no rollback protocol naturally. Every part can be executed concurrently and independently. We take a decentralized execution approach proposed by Anjana\cite{anjana2018efficient} which is shown in Figure \ref{img:decentralizedway}. Multiple threads are working on sub-graphs concurrently in the absence of a master thread.

\begin{figure}[htbp]
	\centering
	\vspace{-0.1cm}  
	\setlength{\abovecaptionskip}{0.1cm}   
	\setlength{\belowcaptionskip}{-0.15cm}   
	\includegraphics[scale=0.56]{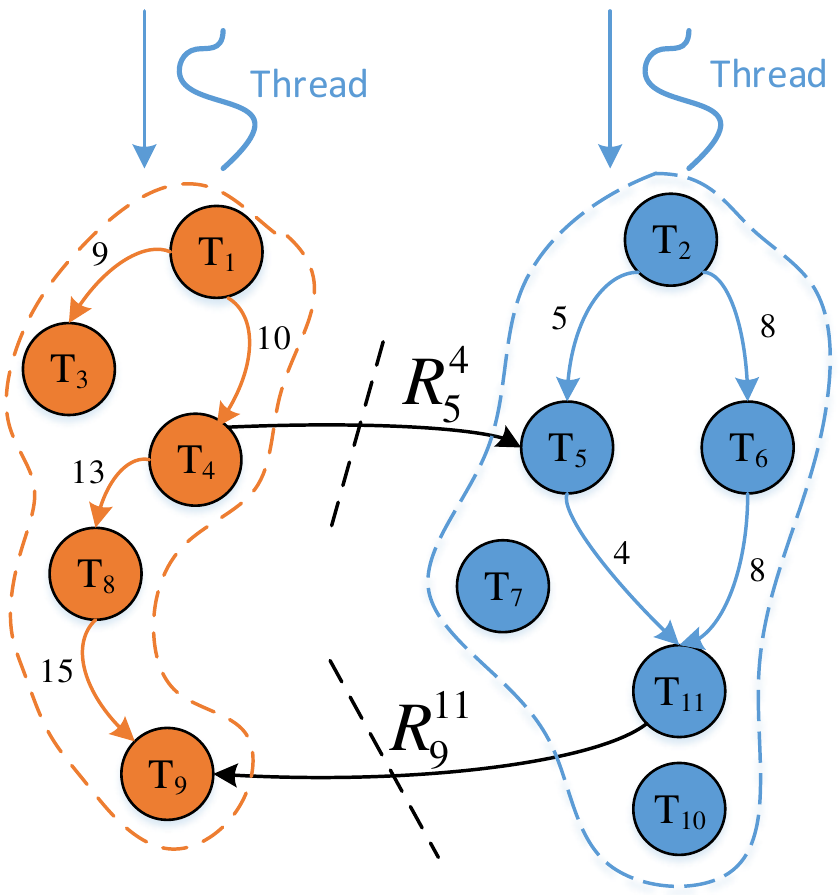}
	\caption{An example of how validators execute transactions in a decentralized way.}
	\label{img:decentralizedway}
\end{figure}

\begin{figure*}[htbp]
	\centering
	\subfigure[Original $TDG$]{
		\label{img:partition:a}
		\includegraphics[scale=0.46]{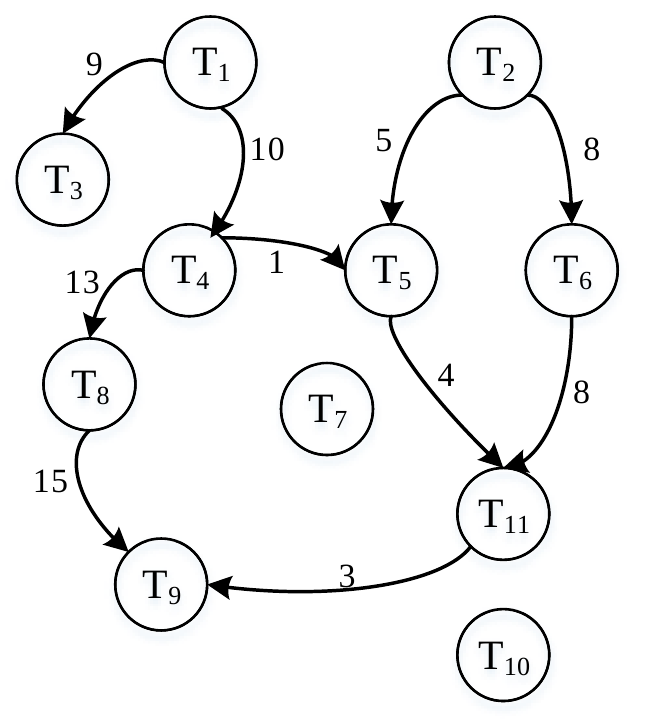}
	}
	\hfill
	\subfigure[Add $T_8$ and $T_9$ to $V_1$]{
		\label{img:partition:b}
		\includegraphics[scale=0.46]{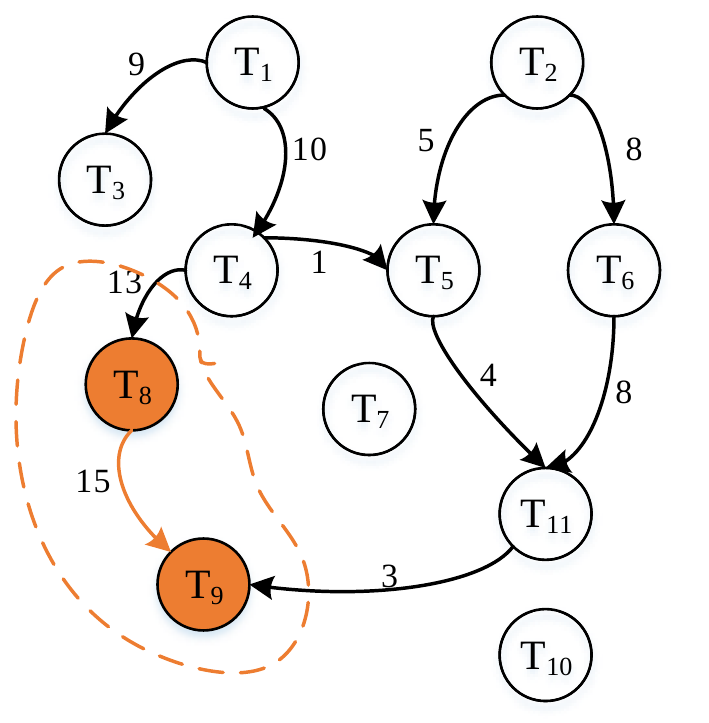}
	}
	\hfill
	\subfigure[Get first part]{
		\label{img:partition:c}
		\includegraphics[scale=0.46]{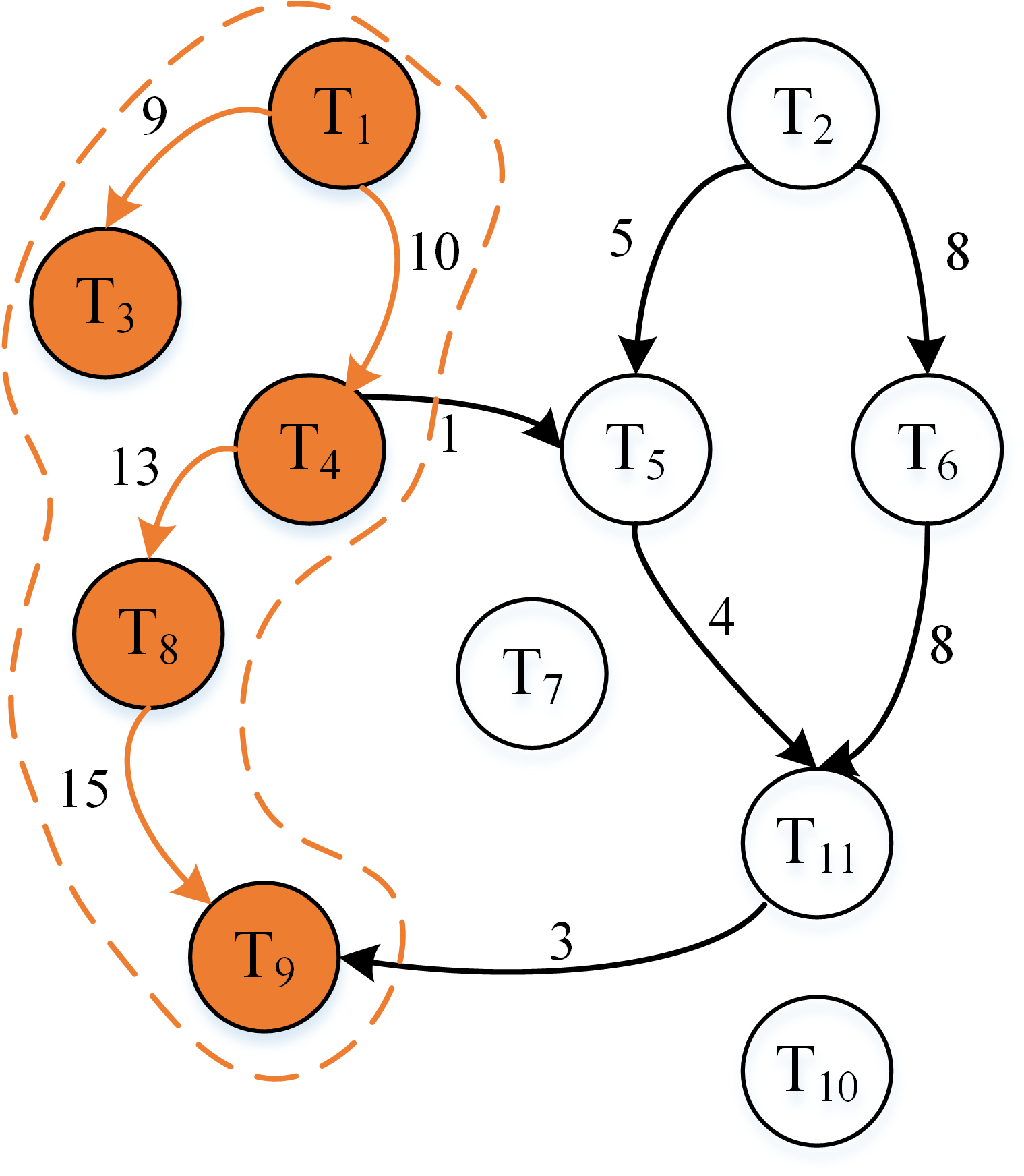}
	}
	\hfill
	\subfigure[Add $T_2$ and $T_6$ to $V_2$]{
		\label{img:partition:d}
		\includegraphics[scale=0.46]{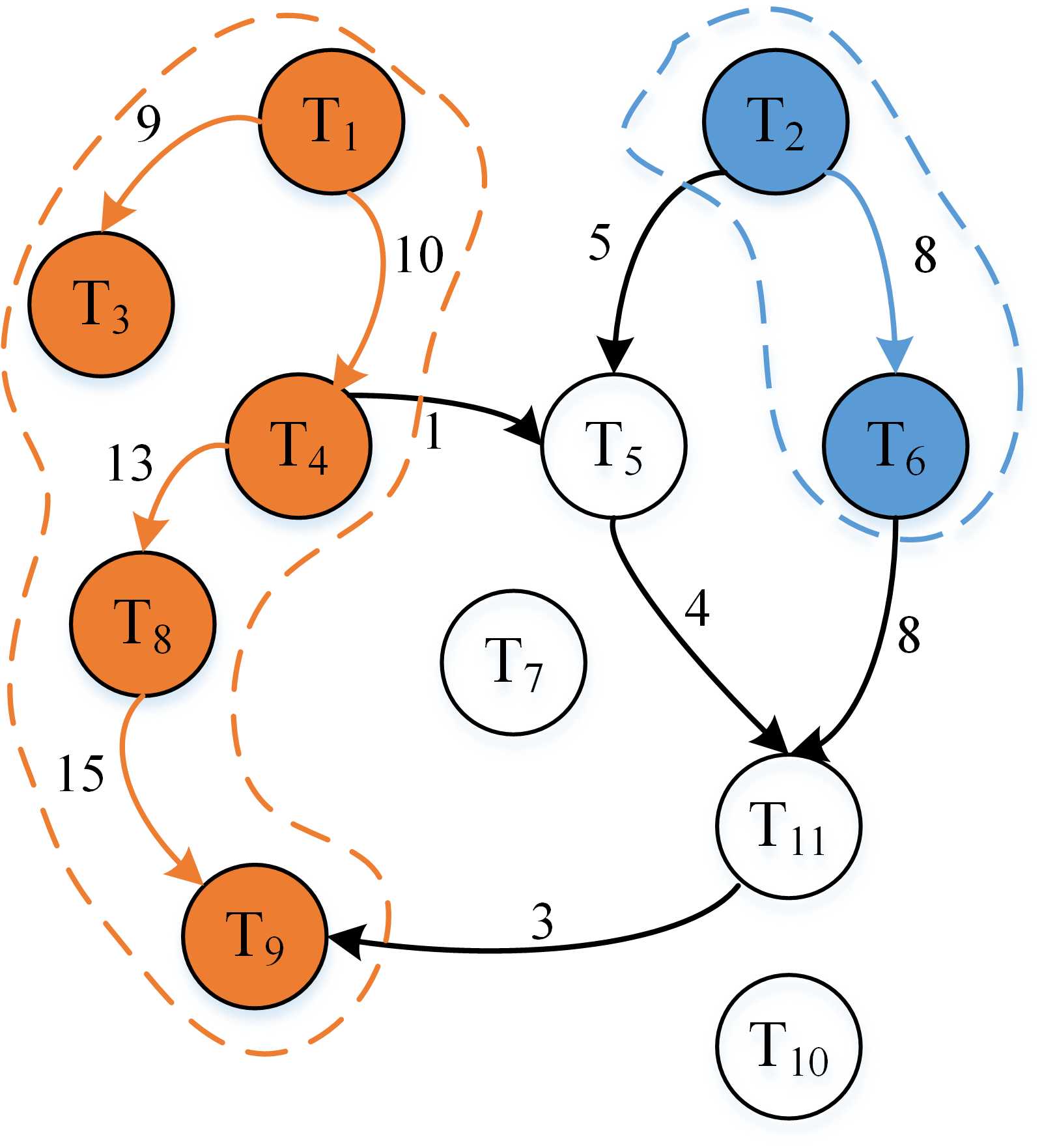}
	}
	\hfill
	\subfigure[Edge cut is 4]{
		\label{img:partition:e}
		\includegraphics[scale=0.46]{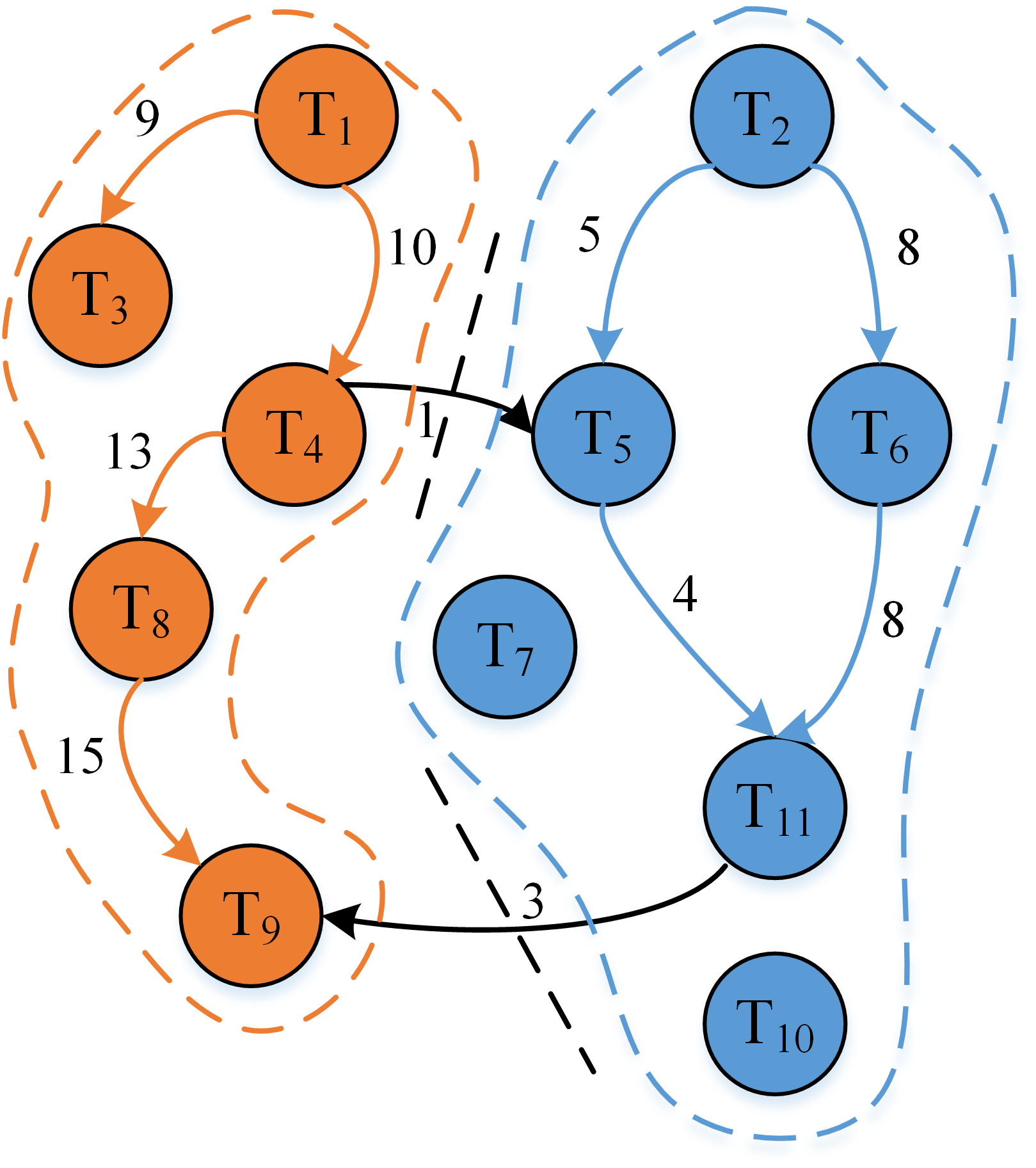}
	}
	\caption{An example of 2-way partitioned $TDG$.}
	\label{img:partition}
\end{figure*}

We add several modifications to allow transactions run concurrently with a predetermined serialization order. Original OCC transactions read values either from shared states or their own write sets. Since a \textit{consistent read set} $R_j^i$ preserves all data item that transaction $T_j$ needs to read from transaction $T_i$'s update, every read operation is able to obtain a consistent value from the graph structure. Although DeOCC transactions have no more occasion to check the consistency of their read sets, they need to verify the correctness of $R_j^i$ cause miner could be a malicious one and send a false value. Of course, the inherent validation mechanism of current blockchain platforms like Ethereum still works under our deterministic protocols. After executing all transactions and finishing the state transition, one can recalculate the state merkle root and compare it with the root sent by the miner. This default setting will waste computational resources for it has to verify the correctness after completing the execution. We propose a new validation scheme embedded in our DeOCC protocol which is explained later. Recall that validators need to produce the same serial order schedule with the miner. Given this constraint, DeOCC transactions must commit according to $\mathcal{O}$. 


Algorithm \ref{alg:deocc} describes a normal DeOCC transaction. When a DeOCC transaction $T$ begins, it gets an additional parameter, serial number $seq_t$, representing the order of this transaction in the serialization order $\mathcal{O}$ discovered by the miner. 


In read phase, $T$ finds a proper value of the data object at most of the time by scanning the \textit{consistent read set}. The reason transactions needn't to check the consistency of its read set is that $R_j^i$ already does that verification by recording every version, i.e., transactions have to read from their former transactions' committed update linked by dependency edges. A read operation on the data object \textit{item} either returns the value from transaction $T$'s \textit{consistent read set} (Line 7) or its own write set (Line 5). A value from local storage is returned in case transaction $T$ doesn't have a \textit{consistent read set}. Write operations remain the same with the OCC transaction which buffer the new value in other locations and apply a \textit{write back} strategy after passing the new verification mechanism. 





Transactions that pass our new validation scheme will enter the write phase. We force transactions to commit according to the predefined order $\mathcal{O}$ so that every transaction has to wait until all its predecessor transactions commit. After writing back transaction's updates, the global sequence value $seq_c$ is assigned to $seq_t$ as line 16 in Algorithm \ref{alg:deocc} suggests.

\begin{algorithm}[ht]
	\caption{A DeOCC transaction $t$}
	\label{alg:deocc}
	{\small
		\When{$txStart(t, seq)$}{
			$seq_t \leftarrow seq$;\\
		}{}
		\When{$txRead(t, item)$}{
			\eIf{$WS(t).contains(item)$}{
				$ReadFromWriteSet(item, WS(t), RS(t))$;}{
				$ReadFromTDG(item, TDG, R_t)$;}
		}{}
		\When{$txWrite(t, item, value)$}{
			$DeferredWrite(item, value, WS(t))$;\\
		}{}
		\When{$txCommit(t)$}{
			\textbf{Serial} /*one by one*/\\
			$ConditionWait(seq_c = predecessor(seq_t))$;\\
			\eIf{$VerifyReadset(RS(t))$}{
				$WriteBack(WS(t))$;\\
				$ComputeConsistentSet(WS(t), TDG)$;\\
				$seq_c \leftarrow seq_t$;\\
			}{
				$AbortTransaction(t)$;\\
			}
		}{}
	}
\end{algorithm}

Restricting commit order directly results in a serial commit which causes a loss of parallelism. To reduce the loss as much as possible, we see room for the transaction that is next to commit. Notice that, at any time, there exists one single transaction which is about to commit according to the predefined commit order. Another fact is that most OCC overheads lie on the deferred write mechanism and validation of the read set. Base on these two facts, we further improve the basic DeOCC protocol with a fast mode proposed. When current committed sequence number $seq_c$ equals to a transaction $T$'s sequence number $seq_t$ (Line 1), we say $T$ then switches to the fast mode. Algorithm \ref{alg:fastdeocc} demonstrates some details about fast mode execution.

\begin{algorithm}[htbp]
	\caption{Fast mode of a DeOCC transaction $t$}
	\label{alg:fastdeocc}
	{\small
		\When{$seq_c = predecessor(seq_t)$}{
			\eIf{$VerifyReadset(RS(t))$}{
				$WriteBack(WS(t))$;\\
			}{
				$AbortTransaction(t)$;\\}
		}{}
		\When{$txRead(t, item)$}{
			$read(t, item)$;\\
		}{}
		\When{$txWrite(t, item, value)$}{
			$DirectUpdate(t, item, value)$;\\
		}{}
		\When{$txCommit(t)$}{
			$ComputeConsistentSet(WS(t), TDG)$;\\
			$seq_c \leftarrow seq_t$;\\
		}{}
	}
\end{algorithm}

As shown in Algorithm \ref{alg:fastdeocc}, the validation phase is moved forward when a transaction is going to enter the fast mode (Line 2-5). Still, transactions need to validate the normal DeOCC execution done up to that point. 

With our fast mode, write operations in read phase use \textit{direct write} strategy instead of deferred updates. Since transactions perform in place updates, read operations no longer read values from $R_j^i$ or its write set. They simply read the current data item's value. Also, the cost of tracking read sets will be eliminated. For we combine read phase and write phase, there is no more \textit{write back} step when transactions commit.

\stitle{Verification of malicious miner}

As we mentioned before, a DeOCC transaction will never read an inconsistent value for $R_j^i$ storing all proper values that transaction needs to read. However, parties involved in blockchain platforms like Ethereum have no trust in each other. In other words, \textit{consistent read sets} transferred by miner require to be verified for their correctness. 

The default setting of blockchain uses root of the Merkle tree to check the validity of the state transition. All of the existing research works adopt the default verification mechanism to detect malicious behaviors of the miner. By this way, however, lots of computational resources may be wasted if the miner is a malicious one. We intend to embed the verification process to our deterministic protocol. Here is how we do the verification. Basically, each transaction computes the \textit{consistent read sets} $R_j^i$ when committing (Line 15 of Algorithm \ref{alg:deocc}). If it provides a read set $\overline{R_j^i}$ in the transaction dependency graph, then the transaction contrasts $R_j^i$ and $\overline{R_j^i}$, and a \textit{verified} tag will be attached to all items in $\overline{R_j^i}$ if $R_j^i = \overline{R_j^i}$. In this way, transactions check whether their read values are all tagged with \textit{verified}. Any unverified item will cause aborting the block. Different from the previous work, our verification scheme can detect malicious miner beforehand.
%
\section{Performance analysis}\label{sec:analysis}
\begin{figure*}[htbp]
	\centering
	\vspace{-0.1cm}  
	\setlength{\abovecaptionskip}{0.1cm}   
	\setlength{\belowcaptionskip}{-0.15cm}   
	\subfigure[Low data contention]{
		\label{img:evaluation:one:a}
		\includegraphics[scale=0.23]{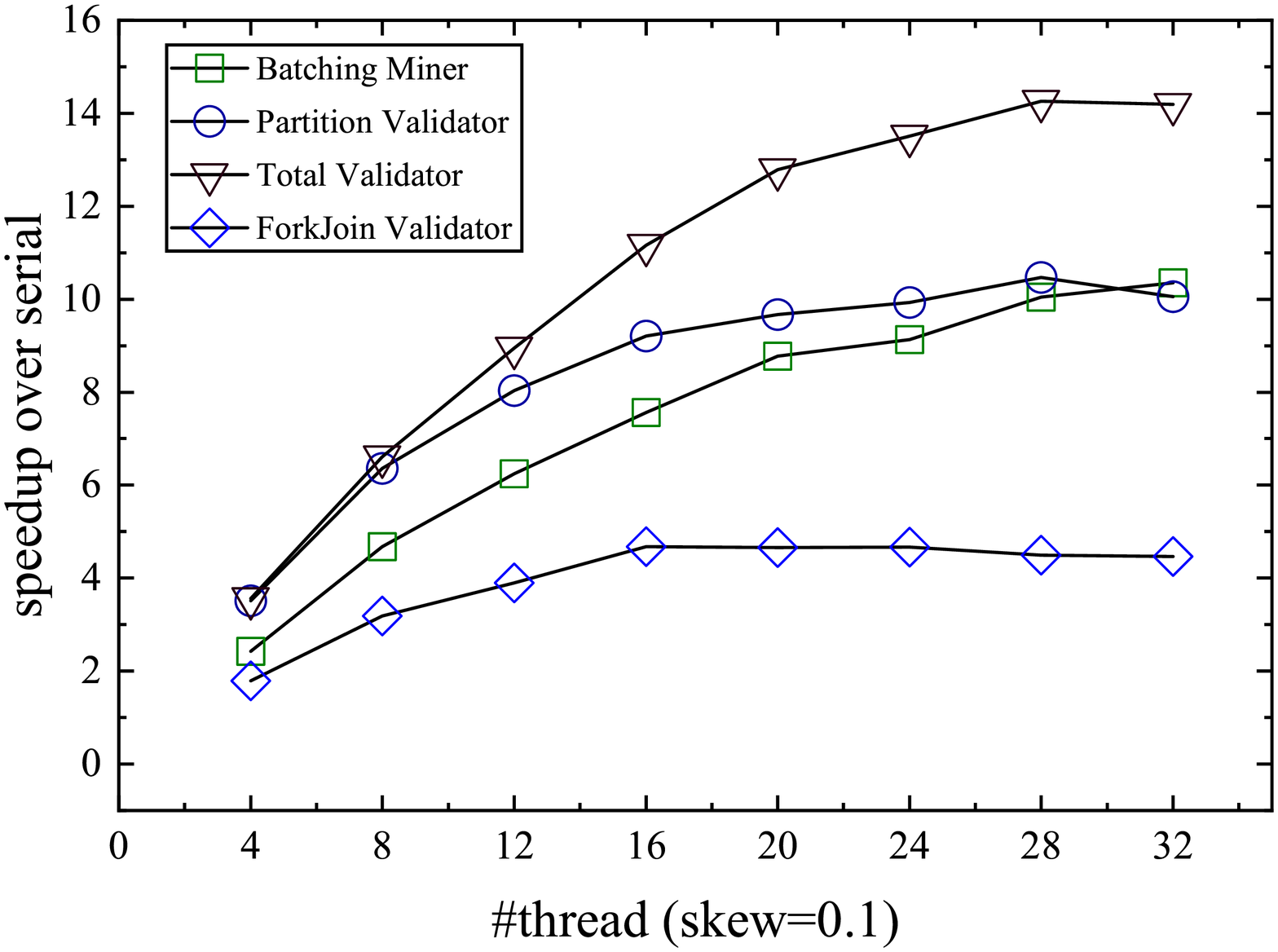}
	}
	\subfigure[Medium data contention]{
		\label{img:evaluation:one:b}
		\includegraphics[scale=0.23]{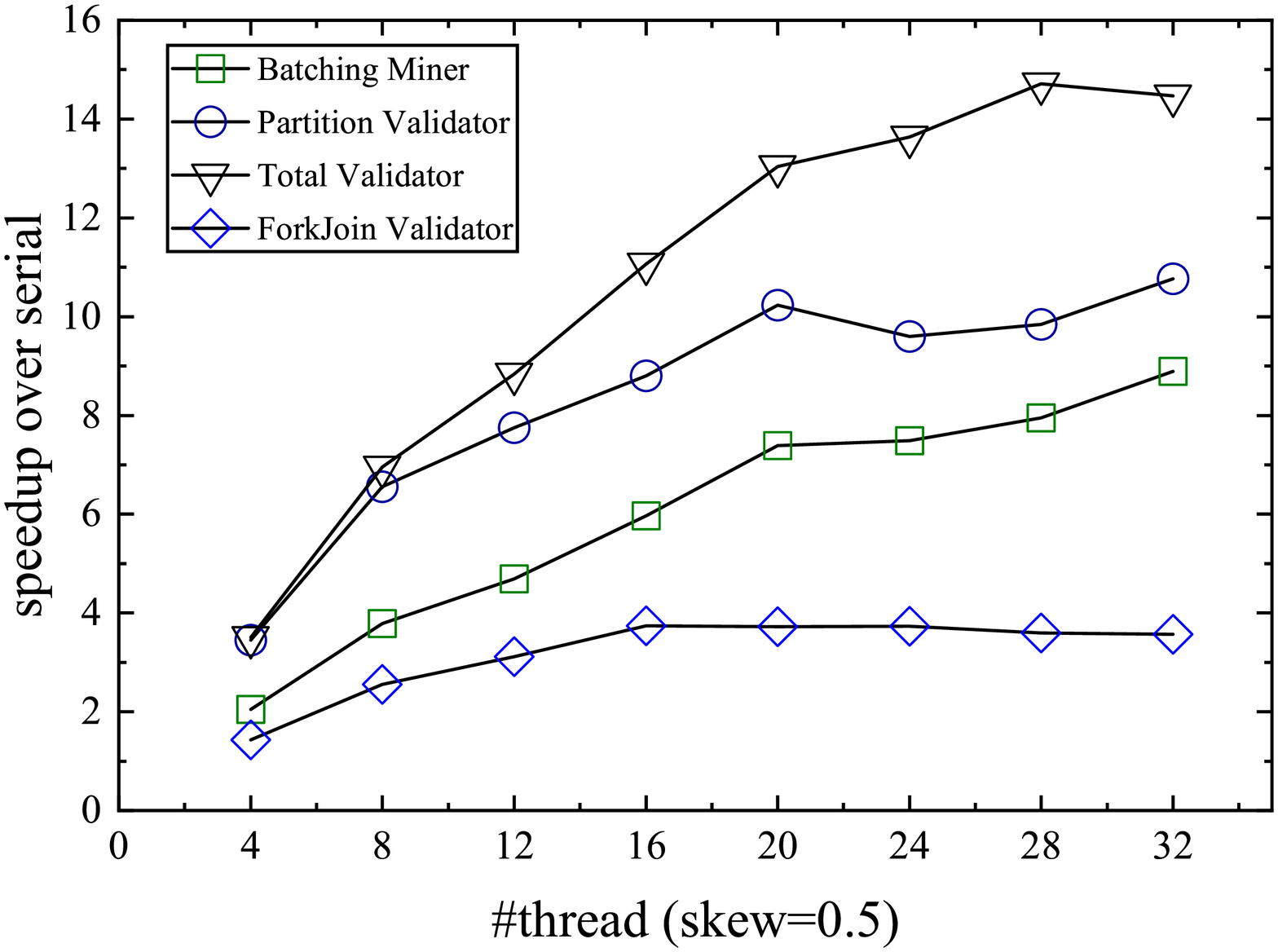}
	}
	\subfigure[High data contention]{
		\label{img:evaluation:one:c}
		\includegraphics[scale=0.23]{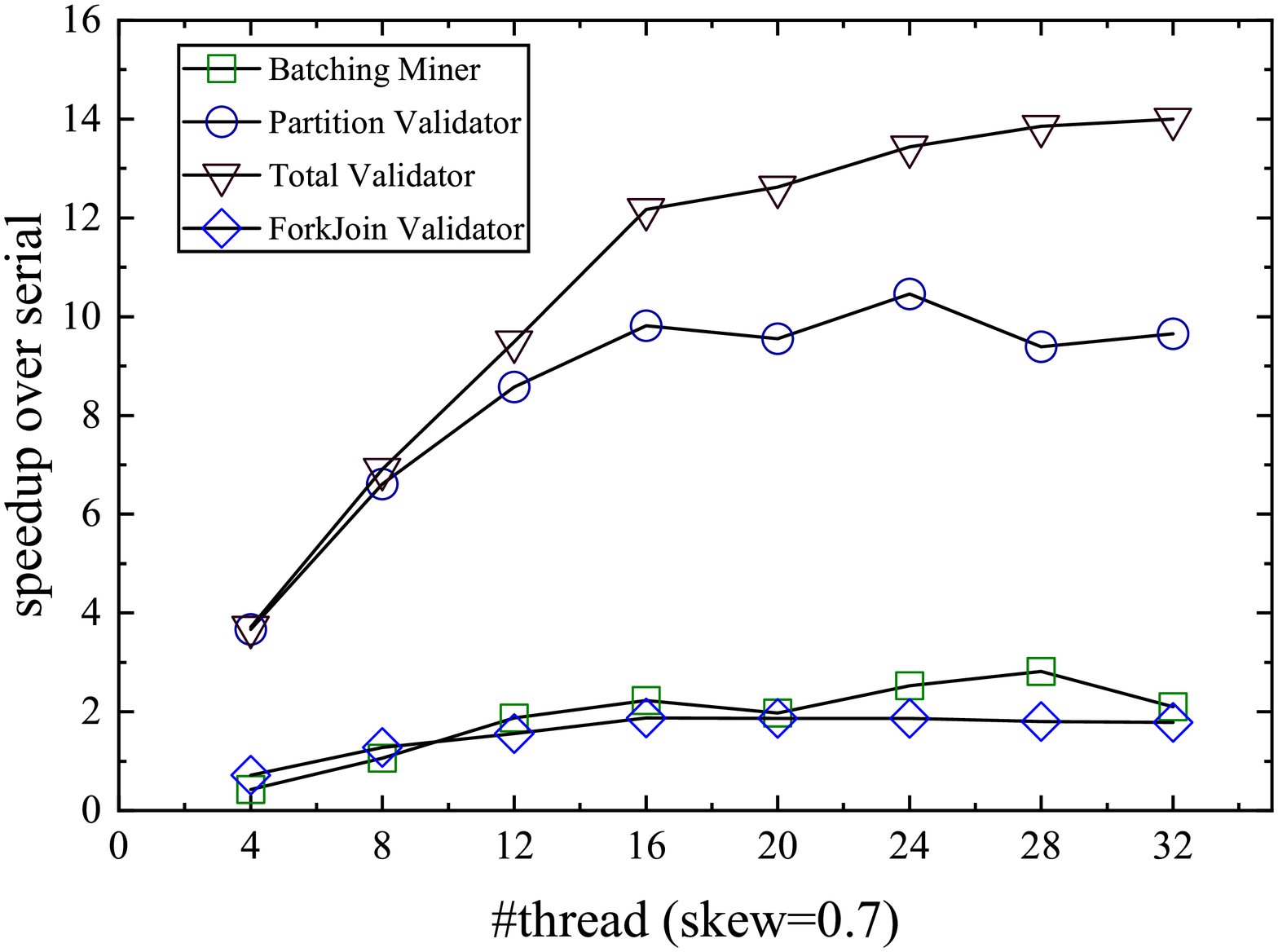}
	}
	\caption{Speedup against the number of threads}
	\label{img:evaluation:one}
\end{figure*}

\begin{figure*}[htbp]
	\centering
	\vspace{-0.15cm}  
	\setlength{\abovecaptionskip}{0.1cm}   
	\setlength{\belowcaptionskip}{-0.05cm}   
	\setlength{\floatsep}{-0.1cm}
	\subfigure[$skew = 0.1$]{
		\label{img:evaluation:two:a}
		\includegraphics[scale=0.23]{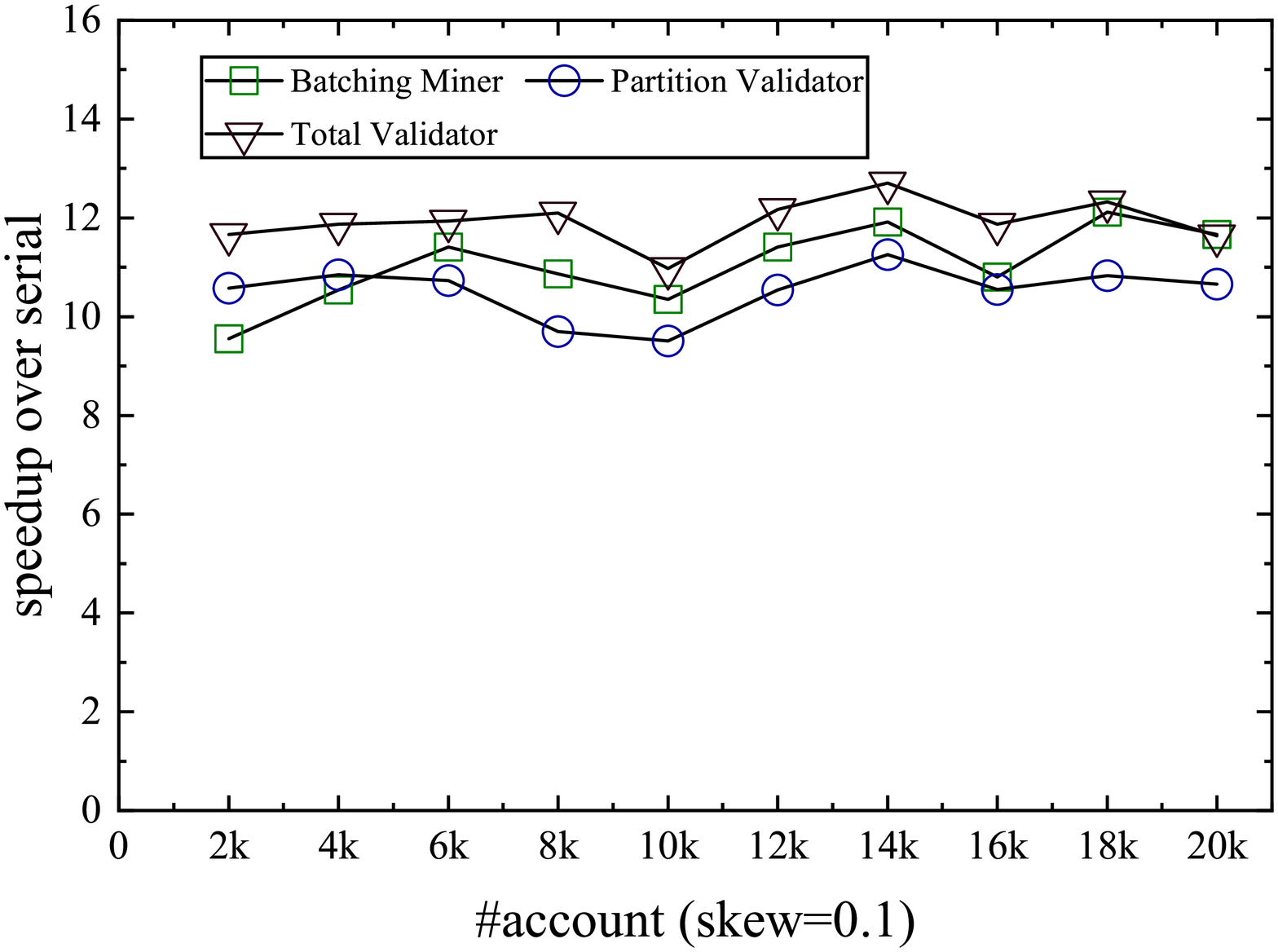}
	}
	\subfigure[$skew = 0.5$]{
		\label{img:evaluation:two:b}
		\includegraphics[scale=0.23]{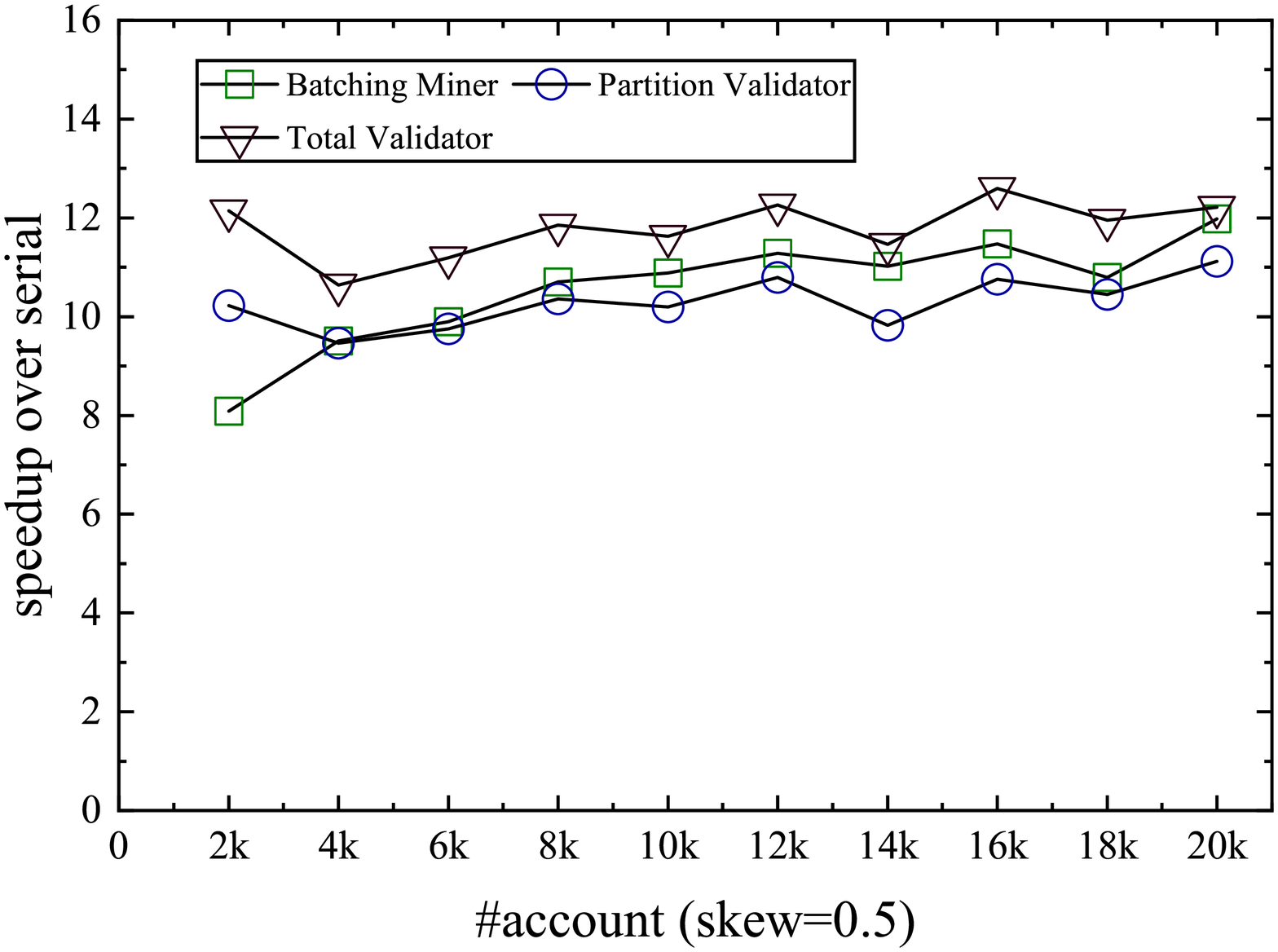}
	}
	\subfigure[$skew = 0.7$]{
		\label{img:evaluation:two:c}
		\includegraphics[scale=0.23]{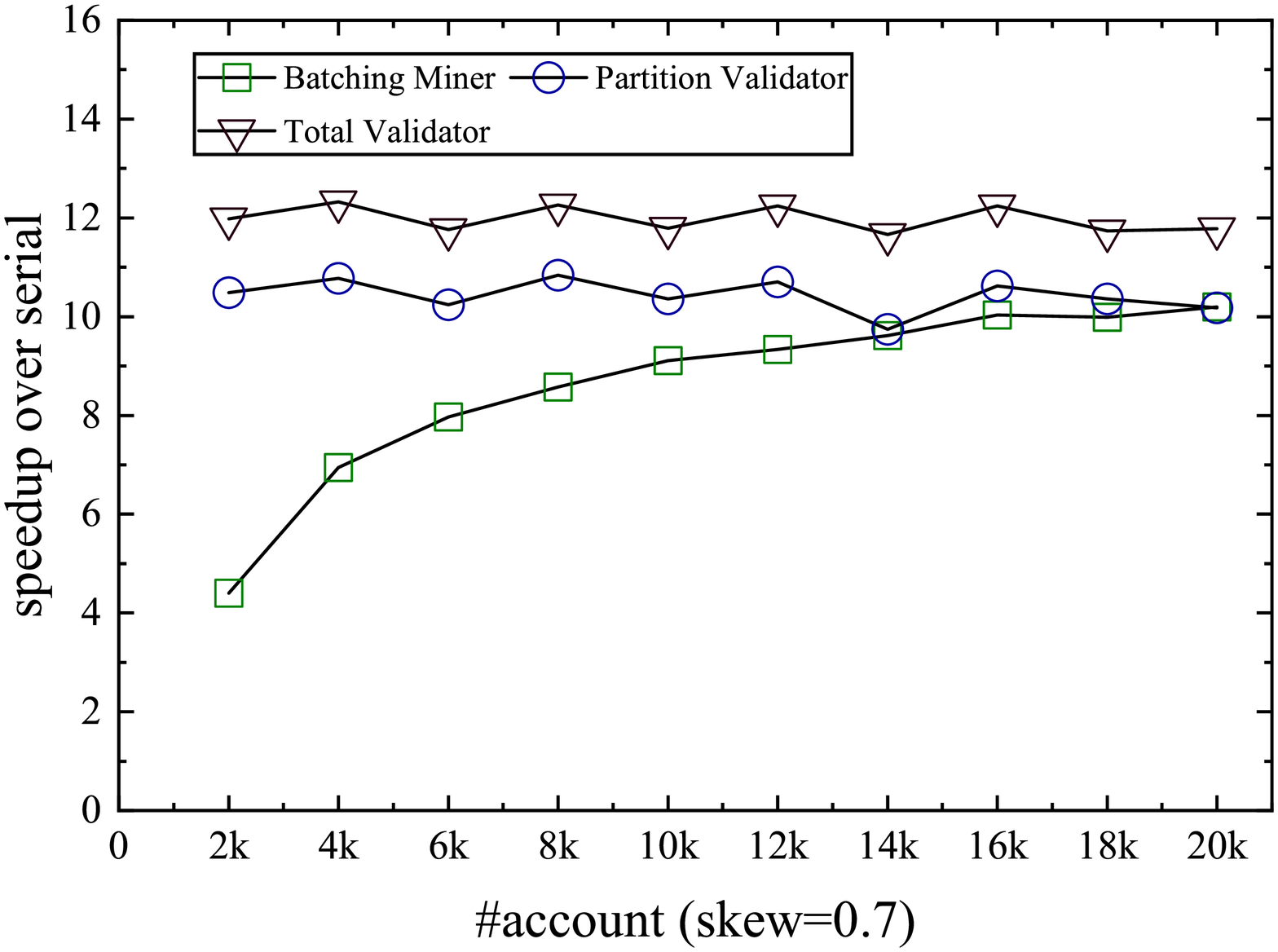}
	}
	\caption{Speedup against the number of accounts}
	\label{img:evaluation:two}
\end{figure*}

We analyze the performance of our two-phase framework in this section.
As mentioned above, miner finds a serialization order using batching OCC in the first phase, and then validators replay and verify the schedule log with DeOCC protocol in the second phase. 
Obviously, the overall cost of our method consists of two parts: communication cost and computational cost.
Communication cost is brought by the information sent from the miner to validators, and the computational cost is the execution time of a batch of transactions. 
To some extent, there is a trade-off between communication cost and computational cost, i.e, receiving more information contained in the schedule log will enable a higher execution speed in validators. 

\stitle{Communication cost}.
The Communication cost refers to the scheduling logs sent by the miner. Suppose the original transaction dependency graph $G=(V,E)$ where $|V|=n$ is broke down to $k$ sub-graphs which are $G_1$, $G_2$, \ldots, $G_k$. Each sub-graph $G_i$ has an extra input data set $D_i$ respectively. The overall communication cost is computed as $|T|=\sum_{i=1}^{k}(|G_i|+|D_i|)$. 

For simplicity, the time to transfer data via network is computed as the data volume times network bandwidth. Let $\alpha$ denote the network bandwidth, the overall cost is computed as $\alpha\cdot |T|$.

\stitle{Computational cost}.
In our framework, all smart contracts must be evaluated one by one in each validator. Moreover, all smart contracts belonging to the same sub-graph will be verified in the same core. Hence, if the number of sub-graphs is smaller than the number of cores, i.e, $k<m$, some cores will be free. Hence, in real-life cases, we will set $k\gg m$ so that all sub-tasks can be assigned to each core evenly to reduce the execution cost. 

\begin{theorem}
	The upper and lower bounds of the computation cost to execute $G$ on an $m$-core processor are computed below, where $\Psi$ is the overall computation cost, and $\tau$ is the maximum size of each sub-graph.
	\begin{eqnarray}
	LB & = & \Psi/m \nonumber\\
	UB & = & \Psi(\frac{1-\tau}{m}+\tau) \nonumber
	\end{eqnarray}
\end{theorem}

\begin{proof}
	Each validator will try to evaluate all smart contracts in a batch. Hence, under the best situation, the evaluation task will be dispatched to $m$ cores evenly. Hence, the lower bound will be $\Psi/m$. However, due to load unbalance, all $m$ cores may not stop at the same time. In the ultimate situation, when one core start to execute the largest sub-task, the rest $m-1$ cores happen to finish all subtasks assigned to them at the same time. In this way, we get the upper bound (UB), as listed above. 
\end{proof}
\vspace{-0.15cm}
It is interesting to discuss the case where each sub-graph only contains one smart contract, i.e, $k=n$. In this case, the $m$-core validate can verify all smart contracts almost evenly. Meanwhile the read set for each smart contract must be ready, which means more data to be transferred. 



%
	\vspace{-0.15cm}
\section{Experiments}\label{sec:evaluation}


In this section, we conduct extensive experiments to evaluate the performance of our proposed method by varying the number of threads, conflict rate, the number of transactions and workload threshold. 

\subsection{Benchmark}
As a popular benchmark for OLTP workload, SmallBank \cite{Cahill2008SerializableIF} is also widely used for blockchain systems \cite{dinh2017blockbench}. Based on a schema of three tables, SmallBank defines four basic procedures to model a simple banking scenario, namely \textit{Amalgamate}, \textit{WriteCheck}, \textit{DepositChecking} and \textit{TransactSaving}, each owning several read/write operations. Since transferring money between accounts is also common in blockchain applications, We add \textit{SendPayment} transaction to extend the original SmallBank benchmark. In real blockchain systems like Ethereum\cite{ethereum}, smart contracts can be treated as transactions that consist of a series of read/write operations. Miner collects transactions calling different smart contracts. So our extended benchmark, called SmallBank+, simulates the real workload well and is appropriate for evaluation. 


\subsection{Experiment Setup}
As the smart contracts written in Solidity upon Ethereum platform are not multithreaded, we implement SmallBank+ in Java to utilize CPU resource more efficiently and evaluate the performance on one machine. Note that this evaluation mechanism is also widely adopted by a series of related works \cite{anjana2018efficient,dickerson2017adding,zhang2018enabling}. Each block contains a set of transactions which are implemented by using \textsf{callable} objects in Java. Our workload generator generates blocks with a combination of transaction count, number of accounts and an initial balance for each group of experiment. The transaction type is generated uniformly using \textsf{Random} class in Java. Data access follows a Zipfian distribution to simulate data skew situations. Specifically, lager skew parameter means that fewer data objects are accessed with higher probability, i.e., more transaction conflicts. A thread pool that is created with \textit{ThreadPoolExecutor}, executes all transactions concurrently both for miner and validators. Serial executor runs with only one thread as a baseline to highlight the effect of our proposed two-phase scheme. We populate the database with 100k customers, including 100k checking and 100k savings accounts. Each group of experiments runs with three different skew parameters indicating the conflict intensity. Since following the similar two-phase execution style with Dickerson's work, we additionally implement the fork-join validator proposed by Dickerson\cite{dickerson2017adding} and evaluate it with SmallBank+ benchmark under the same configuration with our scheme for comparison.

Experiments are conducted on a machine of 160GB memory and a 2-socket Intel Xeon Silver 4110 CPU @2.10GHz with 8 cores per socket and two hyper-threads per core. This machine runs CentOS 7 system with JDK version 1.8. All our experimental figures show the averages of 10 runs. 


\subsection{Experimental Reports}
\begin{figure*}[htbp]
	\centering
	\vspace{-0.1cm}  
	\setlength{\abovecaptionskip}{0.1cm}   
	\setlength{\belowcaptionskip}{-0.2cm}   
	\subfigure[$skew= 0.1$]{
		\label{img:evaluation:three:a}
		\includegraphics[scale=0.23]{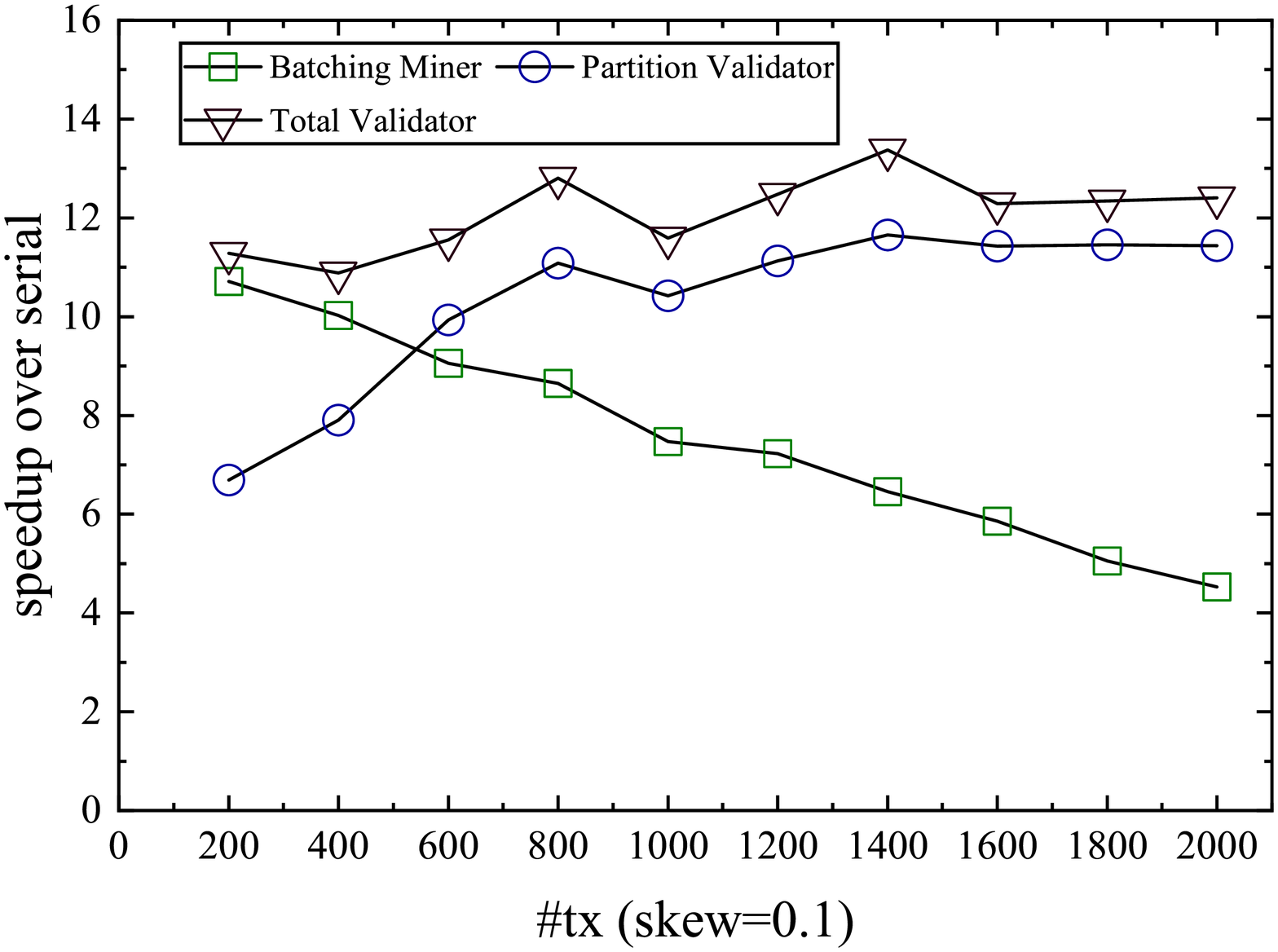}
	}
	\subfigure[$skew=0.5$]{
		\label{img:evaluation:three:b}
		\includegraphics[scale=0.23]{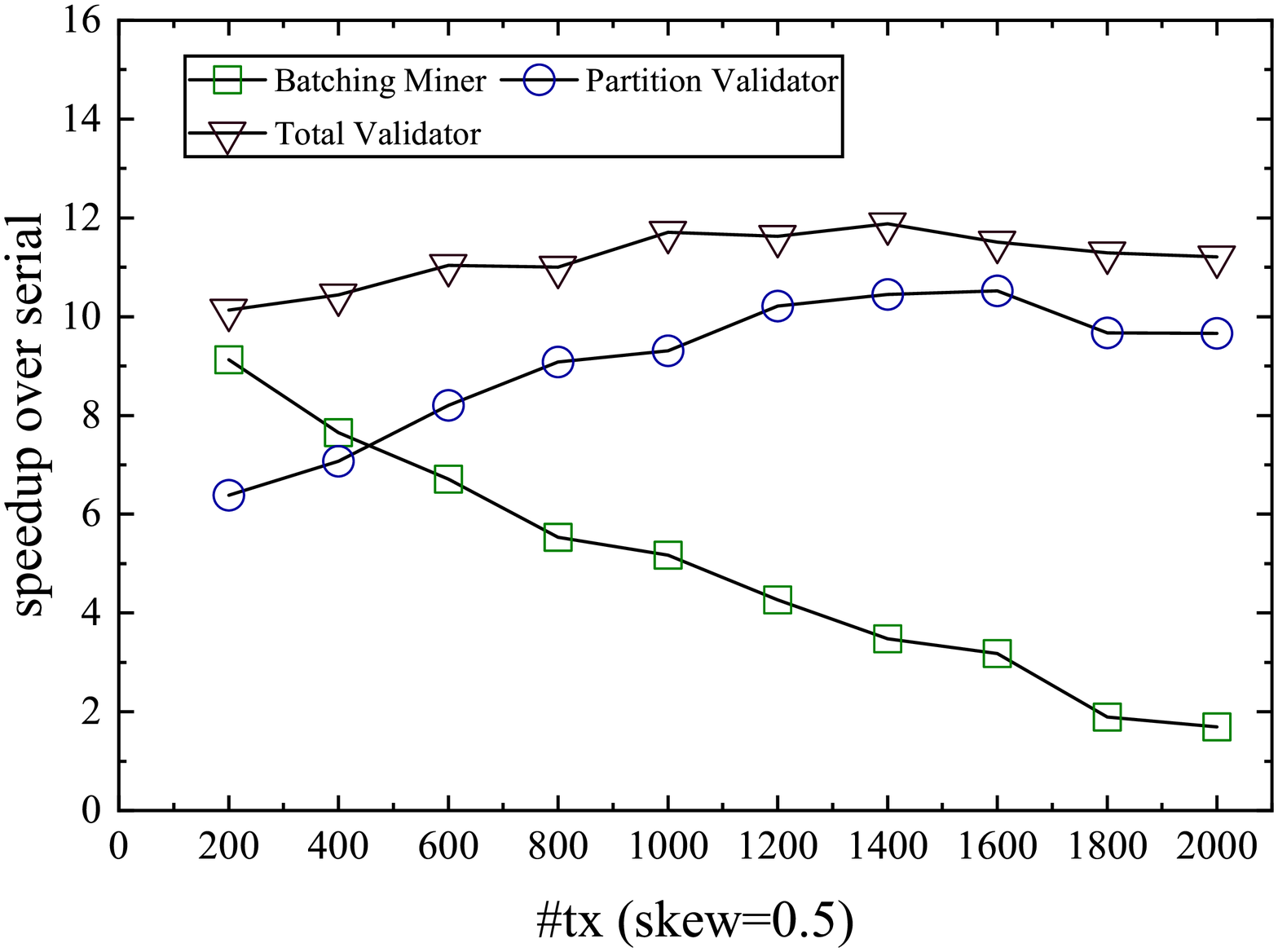}
	}
	\subfigure[$skew = 0.7$]{
		\label{img:evaluation:three:c}
		\includegraphics[scale=0.23]{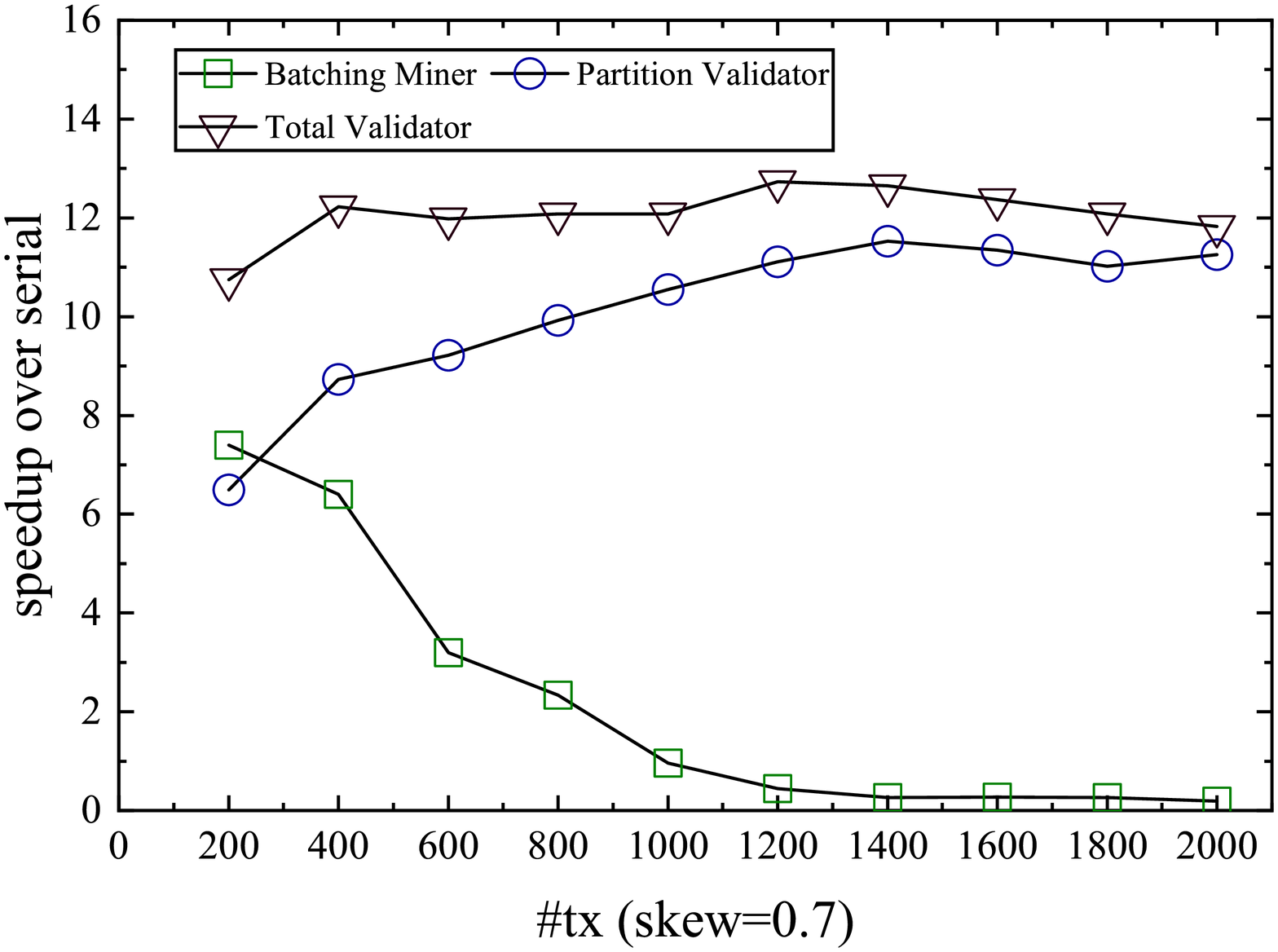}
	}
	\caption{Speedup against the number of transactions}
	\label{img:evaluation:three}
\end{figure*}
%
%

\stitle{Varying the number of threads.}
We evaluate the performance by varying the number of threads from 1 to 32 with a fixed transaction count (400 transactions per block). We evaluate our concurrent scheme and the fork-join validator under three different skew parameters of Zipfian distribution. In reality, conflict is rare because each block contains invocations of different unrelated smart contracts. Figure \ref{img:evaluation:one} shows the speedup over serial execution against the number of threads. Batching Miner refers to the batching protocol used in mining phase. Partition Validator indicates our proposed DeOCC protocol with $\tau$-constrained graph partition while Total Validator means DeOCC without partition. Batching Miner execution achieves approximate 10x speedup when the skew parameter is set to 0.1 (Figure \ref{img:evaluation:one:a}) and 9x speedup even when the workload contains medium rate conflict (skew=0.5). However, the speedup of miner drops to 2x when data skew increases to 0.7, because high data contention will result in too many aborts which are infrequent in blockchain applications. Partition Validator shows a good boost on performance which gains 10x speedup on average. Total Validator has the highest concurrency degree which achieves a maximum of 14x speedup as shown in Figure \ref{img:evaluation:one:b}.

The speedup of Partition Validator in Figure \ref{img:evaluation:one} follows the similar trend which indicates the protocol is resistant for the change of conflict rate. 
Result in Figure \ref{img:evaluation:one:a} shows that the speedup of the fork-join method is below 4. The performance of fork-join method keeps falling down when data skew increases. In all three cases, both batching miner and partition validator outperform the fork-join method.

\stitle{Varying the number of accounts.}
We evaluate the influence on speedup brought by different accounts with a fixed number of threads and 400 transactions per block. We find that the performance of batching miner firstly goes up and finally approaches the speedup of Total Validator (12.5x) when the number of accounts increases, because the conflict varies inversely with the number of accessing accounts. Partition validator performs a speedup of 11x steadily regardless of the skew parameter. The results shown in Figure \ref{img:evaluation:two} also bear out our conclusion that partition validator is rarely affected by the adjustment of data skew. 

\stitle{Varying the number of transactions in one block.}
Our implementation is evaluated under blocks containing between 100 to 2,000 transactions with fixed 16 threads and 1,000 accounts. Figure \ref{img:evaluation:three} shows the speedup against the number of transactions per block. Using batching OCC in miner phase improves the execution speed up to 10x when the number of transactions is lower than 400. In general, the throughput of batching miner decreases when more transactions are included in a block. If data contention is extremely high (e.g., data skew is set to 0.7), the number of aborted transactions in one batch increases markedly. Batching and reordering make the performance worse than serial execution when there are more than 900 transactions, because the conflict graph becomes denser during validation. Partition validator only obtains half of its maximal speedup when each block contains fewer transactions (e.g., 100 transactions). This is because it's more difficult to dispatch tasks evenly to all cores to utilize CPU resources efficiently when transactions are few.

\stitle{Impact of the workload threshold $\tau$ on performance.}
The last two groups of experiments perform a detailed analysis of how much impact on performance and communication cost threshold $\tau$ has. Experiment four varies $\tau$ from 0.0035 to 0.056 with 16 threads running. Figure \ref{img:evaluation:four} shows the speedup against different workload threshold. When threshold $\tau$ goes up, the speedup slightly drops. As reported in Section \ref{sec:analysis}, larger $\tau$ causes more workload being executed serially. When $\tau < 0.02$, the speedup of partition validator and total validator are well-matched, which means our partition algorithm can remain the parallelism of validators as much as possible. According to our cost analysis in Section \ref{sec:analysis}, the performance will be close to the case where we provide a consistent read set for each vertex. 

\begin{figure*}[htbp]
	\centering
	\subfigure[Low contention.]{
		\label{img:evaluation:four:a}
		\includegraphics[scale=0.23]{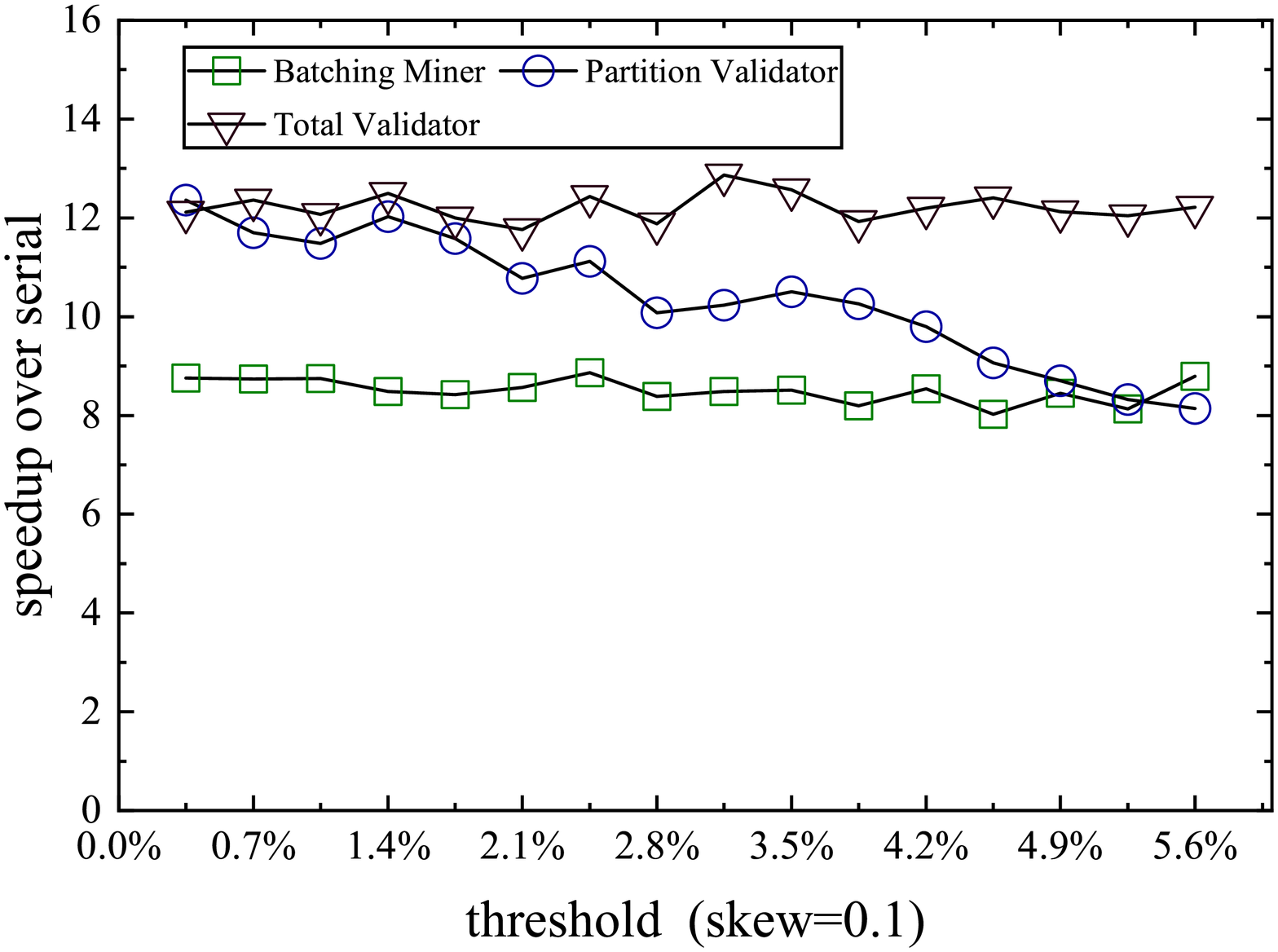}
	}
	\subfigure[Medium contention.]{
		\label{img:evaluation:four:b}
		\includegraphics[scale=0.23]{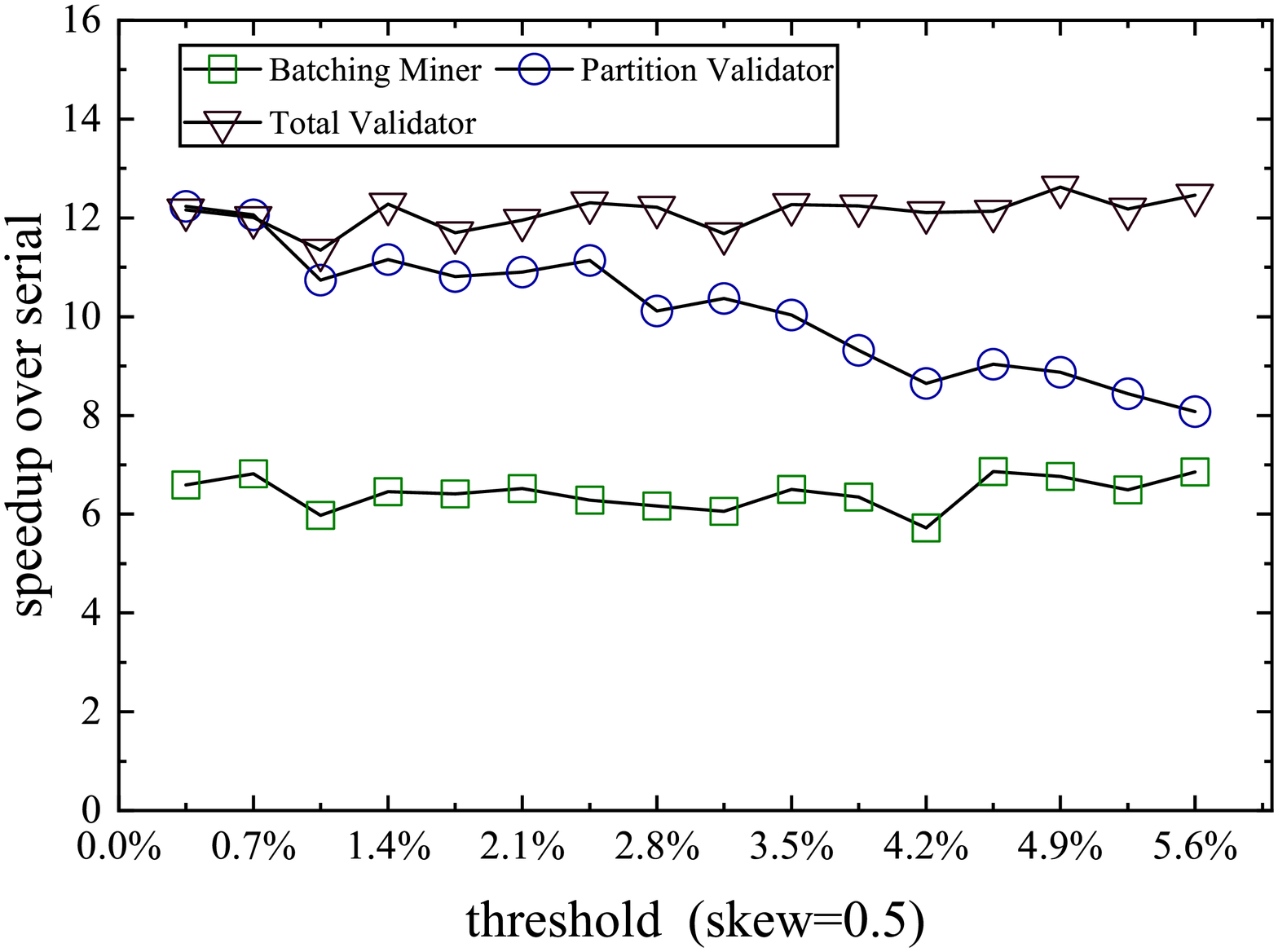}
	}
	\subfigure[High contention.]{
		\label{img:evaluation:four:c}
		\includegraphics[scale=0.23]{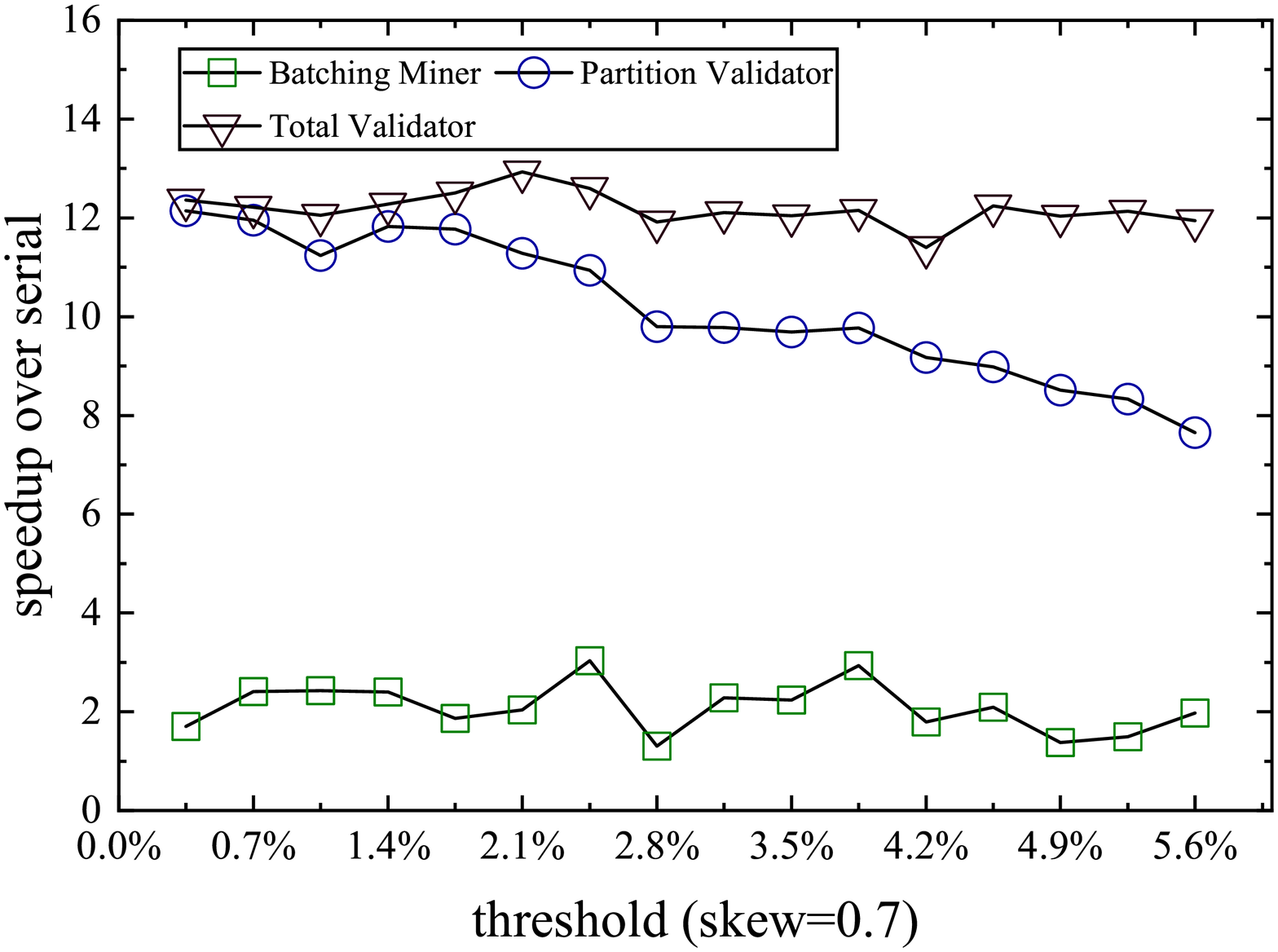}
	}
	\caption{Speedup against the workload threshold ($\tau$)}
	\label{img:evaluation:four}
\end{figure*}

\begin{figure*}[htbp]
	\centering
	\vspace{-0.1cm}  
	\setlength{\abovecaptionskip}{0.1cm}   
	\setlength{\belowcaptionskip}{-0.15cm}   
	\subfigure[$skew= 0.1$]{
		\label{img:evaluation:five:a}
		\includegraphics[scale=0.22]{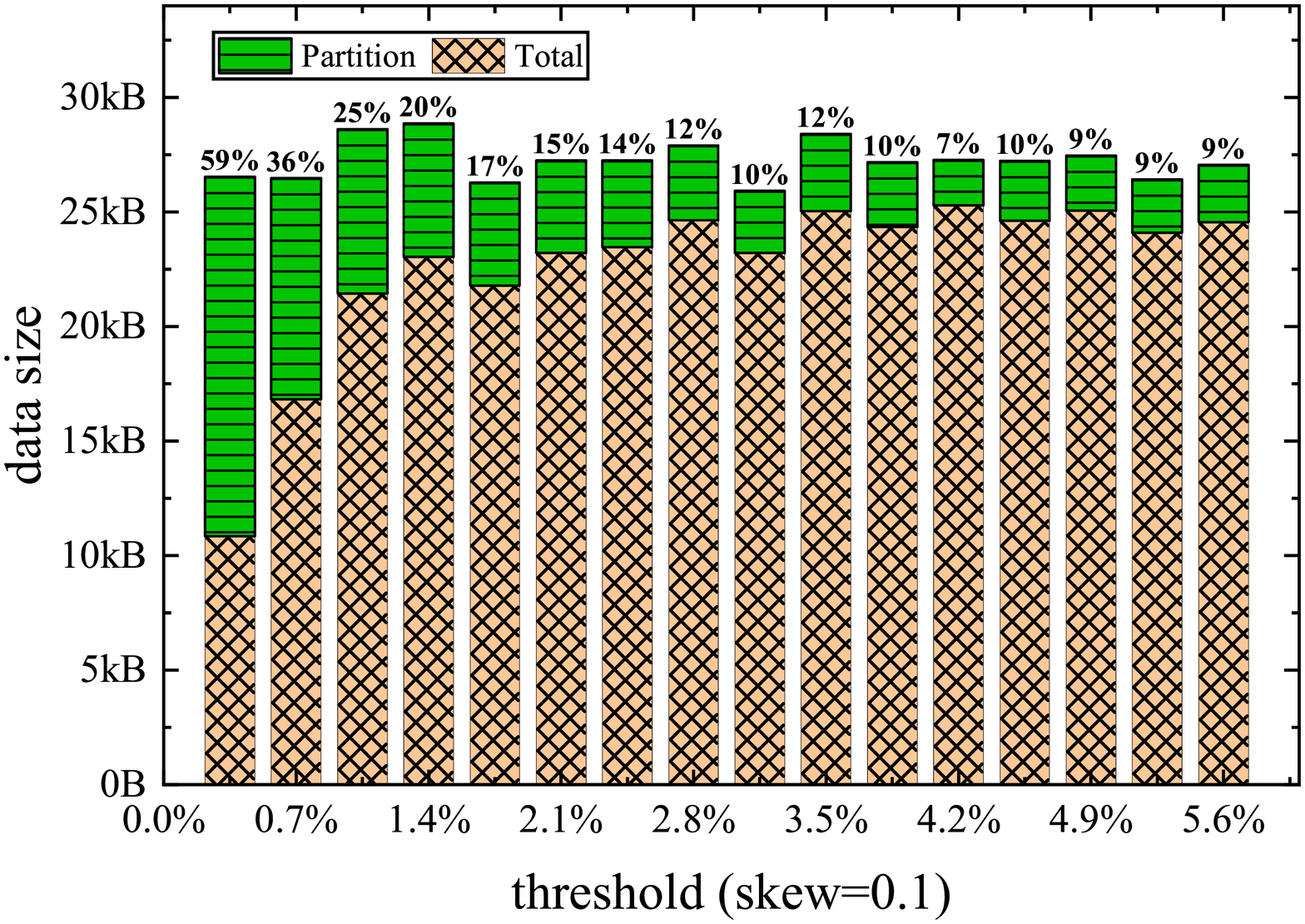}
	}
	\subfigure[$skew= 0.5$]{
		\label{img:evaluation:five:b}
		\includegraphics[scale=0.22]{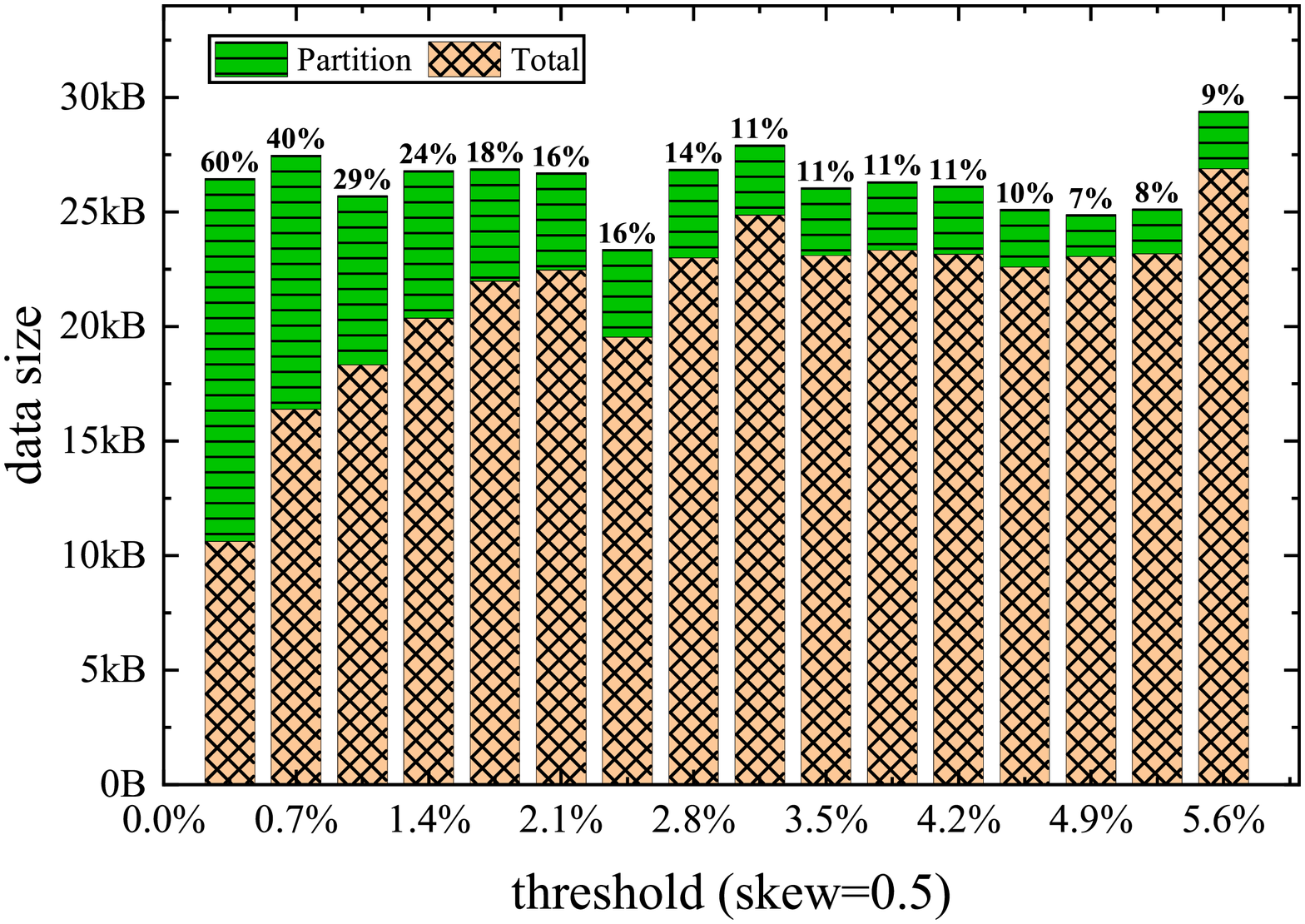}
	}
	\subfigure[$skew=0.7$]{
		\label{img:evaluation:five:c}
		\includegraphics[scale=0.22]{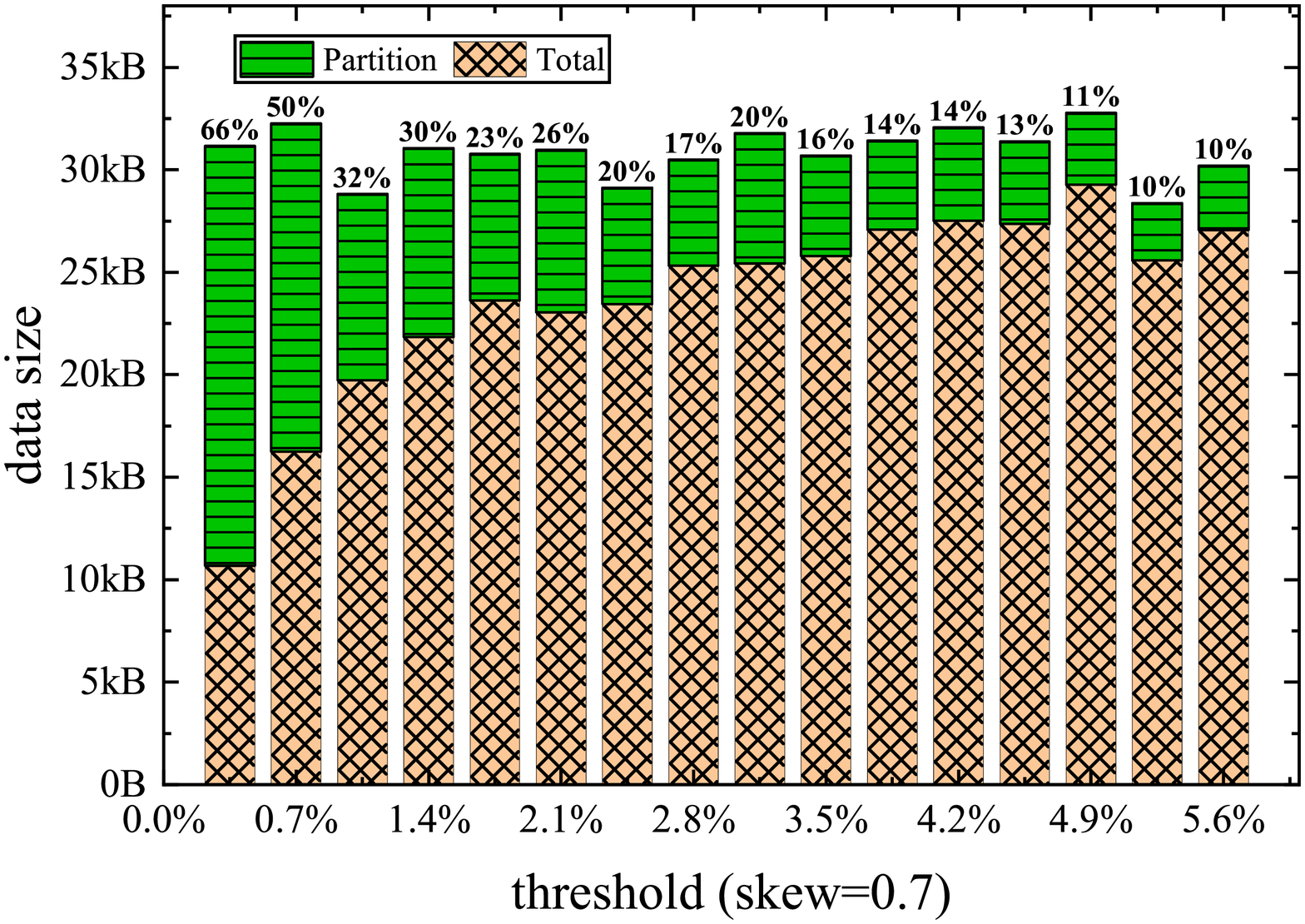}
	}
	\caption{Communication cost with different $\tau$}
	\label{img:evaluation:five}
\end{figure*}

\stitle{The saving of communication cost with varied $\tau$.}
The last experiment evaluate the saving of communication cost with $\tau$ varying from 0.0035 to 0.056. In Figure \ref{img:evaluation:five}, as data skew increases, communication overhead also rises due to more inter-conflicts. In all three cases, our partition algorithm significantly reduces the transferred size of consistent read set when $\tau > 0.015$. Even in the worst case where skew is set to 0.7, the reduction of communication cost is about 85\%.

Overall, our approach achieves an average 9x speedup for the miner and 11x for concurrent validators. And we reduce approximate 90\% communication overhead by sacrificing about 21\% performance.
%
\section{Conclusion}\label{sec:conclusion}


In this study, we present an efficient two-phase execution framework to add concurrency to smart contracts aiming at higher parallelism at both miner and validators.
More specifically, this framework includes three aspects: First, we propose an efficient variant of OCC protocol combined with batching feature of blockchain systems in the mining phase, where a greedy algorithm is devised to solve the FVS problem on the conflict graph. 
Second, we design an appropriate granularity of schedule log, along with a practical partition method to reduce the communication overhead and remain high concurrent degree in the validator phase. Third, we bring up a concurrent scheme, named DeOCC, to deterministically and efficiently replay the same schedule discovered by the miner. The evaluation shows that our two-phase execution framework can achieve approximate 11x speedup both for the miner and validators and outperform state-of-art solutions significantly. Moreover, the communication overhead drops sharply after applying our graph partition algorithm.

There are two possible pieces of future work. The first one is to explore adaptive concurrency control for transactions in blockchain system, which can dynamically fit all varied workloads. The second one is to search for solutions of concurrent execution in TEE (Trusted Execution Environment) represented by SGX (Intel Software Guard Extensions). 
	
	\bibliographystyle{IEEEtran}
	\bibliography{mybiblio}

\begin{thebibliography}{10}
\providecommand{\url}[1]{#1}
\csname url@samestyle\endcsname
\providecommand{\newblock}{\relax}
\providecommand{\bibinfo}[2]{#2}
\providecommand{\BIBentrySTDinterwordspacing}{\spaceskip=0pt\relax}
\providecommand{\BIBentryALTinterwordstretchfactor}{4}
\providecommand{\BIBentryALTinterwordspacing}{\spaceskip=\fontdimen2\font plus
\BIBentryALTinterwordstretchfactor\fontdimen3\font minus
  \fontdimen4\font\relax}
\providecommand{\BIBforeignlanguage}[2]{{%
\expandafter\ifx\csname l@#1\endcsname\relax
\typeout{** WARNING: IEEEtran.bst: No hyphenation pattern has been}%
\typeout{** loaded for the language `#1'. Using the pattern for}%
\typeout{** the default language instead.}%
\else
\language=\csname l@#1\endcsname
\fi
#2}}
\providecommand{\BIBdecl}{\relax}
\BIBdecl

\bibitem{ethereum}
``Ethereum,'' \url{https://github.com/ethereum}, 2014.

\bibitem{hyperledgerfabric}
``Hyperledger fabric,'' \url{https://github.com/hyperledger/fabric}, 2015.

\bibitem{szabo1997formalizing}
N.~Szabo, ``Formalizing and securing relationships on public networks,''
  \emph{First Monday}, vol.~2, no.~9, 1997.

\bibitem{dinh2018untangling}
T.~T.~A. Dinh, R.~Liu, M.~Zhang, G.~Chen, B.~C. Ooi, and J.~Wang, ``Untangling
  blockchain: A data processing view of blockchain systems,'' \emph{IEEE
  Transactions on Knowledge and Data Engineering}, vol.~30, no.~7, pp.
  1366--1385, 2018.

\bibitem{anjana2018efficient}
P.~S. Anjana, S.~Kumari, S.~Peri, S.~Rathor, and A.~Somani, ``An efficient
  framework for concurrent execution of smart contracts,'' \emph{arXiv preprint
  arXiv:1809.01326}, 2018.

\bibitem{dickerson2017adding}
T.~Dickerson, P.~Gazzillo, M.~Herlihy, and E.~Koskinen, ``Adding concurrency to
  smart contracts,'' in \emph{Proceedings of the ACM Symposium on Principles of
  Distributed Computing}.\hskip 1em plus 0.5em minus 0.4em\relax ACM, 2017, pp.
  303--312.

\bibitem{zhang2018enabling}
A.~Zhang and K.~Zhang, ``Enabling concurrency on smart contracts using
  multiversion ordering,'' in \emph{APWeb and WAIM Joint International
  Conference on Web and Big Data}.\hskip 1em plus 0.5em minus 0.4em\relax
  Springer, 2018, pp. 425--439.

\bibitem{lea2000java}
D.~Lea, ``A java fork/join framework,'' in \emph{Java Grande}, 2000, pp.
  36--43.

\bibitem{eswaran1976notions}
K.~P. Eswaran, J.~N. Gray, R.~A. Lorie, and I.~L. Traiger, ``The notions of
  consistency and predicate locks in a database system,'' \emph{Communications
  of the ACM}, vol.~19, no.~11, pp. 624--633, 1976.

\bibitem{kung1981optimistic}
H.-T. Kung and J.~T. Robinson, ``On optimistic methods for concurrency
  control,'' \emph{ACM Transactions on Database Systems (TODS)}, vol.~6, no.~2,
  pp. 213--226, 1981.

\bibitem{wang2016mostly}
T.~Wang and H.~Kimura, ``Mostly-optimistic concurrency control for highly
  contended dynamic workloads on a thousand cores,'' \emph{Proceedings of the
  VLDB Endowment}, vol.~10, no.~2, pp. 49--60, 2016.

\bibitem{sergey2017concurrent}
I.~Sergey and A.~Hobor, ``A concurrent perspective on smart contracts,'' in
  \emph{International Conference on Financial Cryptography and Data
  Security}.\hskip 1em plus 0.5em minus 0.4em\relax Springer, 2017, pp.
  478--493.

\bibitem{ding2018improving}
B.~Ding, L.~Kot, and J.~Gehrke, ``Improving optimistic concurrency control
  through transaction batching and operation reordering,'' \emph{Proceedings of
  the VLDB Endowment}, vol.~12, no.~2, pp. 169--182, 2018.

\bibitem{santos2012tuning}
N.~Santos and A.~Schiper, ``Tuning paxos for high-throughput with batching and
  pipelining,'' in \emph{International Conference on Distributed Computing and
  Networking}.\hskip 1em plus 0.5em minus 0.4em\relax Springer, 2012, pp.
  153--167.

\bibitem{bocchino2009parallel}
R.~L. Bocchino~Jr, V.~S. Adve, S.~V. Adve, and M.~Snir, ``Parallel programming
  must be deterministic by default,'' in \emph{Proceedings of the First USENIX
  conference on Hot topics in parallelism}, 2009, pp. 4--4.

\bibitem{vale2016pot}
T.~M. Vale, J.~A. Silva, R.~J. Dias, and J.~M. Louren{\c{c}}o, ``Pot:
  Deterministic transactional execution,'' \emph{ACM Transactions on
  Architecture and Code Optimization (TACO)}, vol.~13, no.~4, p.~52, 2016.

\bibitem{adya1999weak}
A.~Adya, ``Weak consistency: a generalized theory and optimistic
  implementations for distributed transactions,'' 1999.

\bibitem{cahill2008serializable}
M.~J. Cahill, U.~R{\"o}hm, and A.~D. Fekete, ``Serializable isolation for
  snapshot databases,'' in \emph{Proceedings of the 2008 ACM SIGMOD
  international conference on Management of data}.\hskip 1em plus 0.5em minus
  0.4em\relax ACM, 2008, pp. 729--738.

\bibitem{chen2008improved}
J.~Chen, F.~V. Fomin, Y.~Liu, S.~Lu, and Y.~Villanger, ``Improved algorithms
  for feedback vertex set problems,'' \emph{Journal of Computer and System
  Sciences}, vol.~74, no.~7, pp. 1188--1198, 2008.

\bibitem{kann1992approximability}
V.~Kann, ``On the approximability of np-complete optimization problems,'' Ph.D.
  dissertation, 1992.

\bibitem{karp1972reducibility}
R.~M. Karp, ``Reducibility among combinatorial problems,'' in \emph{Complexity
  of computer computations}.\hskip 1em plus 0.5em minus 0.4em\relax Springer,
  1972, pp. 85--103.

\bibitem{sharir1981strong}
M.~Sharir, ``A strong-connectivity algorithm and its applications in data flow
  analysis,'' \emph{Computers \& Mathematics with Applications}, vol.~7, no.~1,
  pp. 67--72, 1981.

\bibitem{gabow2000path}
H.~N. Gabow, ``Path-based depth-rst search for strong and biconnected
  components,'' \emph{Information Processing Letters}, 2000.

\bibitem{Cahill2008SerializableIF}
M.~J. Cahill, U.~R{\"o}hm, and A.~Fekete, ``Serializable isolation for snapshot
  databases,'' in \emph{SIGMOD Conference}, 2008.

\bibitem{dinh2017blockbench}
T.~T.~A. Dinh, J.~Wang, G.~Chen, R.~Liu, B.~C. Ooi, and K.-L. Tan,
  ``Blockbench: A framework for analyzing private blockchains,'' in
  \emph{Proceedings of the 2017 ACM International Conference on Management of
  Data}.\hskip 1em plus 0.5em minus 0.4em\relax ACM, 2017, pp. 1085--1100.

\end{thebibliography}
\end{document}